\documentclass[journal]{IEEEtran}
\usepackage{amssymb}
\usepackage{comment}
\usepackage[usenames,dvipsnames]{pstricks}
\usepackage{epsfig}
\usepackage{amsmath}
\usepackage{setspace}
\usepackage{subfigure}
\usepackage{enumerate}
\usepackage{bm}
\usepackage{cite}

\allowdisplaybreaks

\newtheorem{definition}{Definition}
\newtheorem{theorem}{Theorem}
\newtheorem{lemma}{Lemma}

\long\def\symbolfootnote[#1]#2{\begingroup%
 \def\thefootnote{\fnsymbol{footnote}}\footnote[#1]{#2}\endgroup}

\def\Z {{\mathbb {Z}}}

\def\R {{\mathbb {R}}}
\def\bx {{\bf x}}	\def\bv {{\bf v}}	\def\bz {{\bf z}}
\def\by {{\bf y}} 
 
\begin{document}

\title{Telescoping Recursive Representations and Estimation of Gauss-Markov Random Fields}
\author{Divyanshu Vats,~\IEEEmembership{Student Member,~IEEE,} and Jos\'{e} M. F. Moura,~\IEEEmembership{Fellow,~IEEE}%
\thanks{To appear in the Transactions on Information Theory}
\thanks{The authors are with the Department of Electrical and Computer Engineering, Carnegie Mellon University, Pittsburgh, PA, 15213, USA (email: dvats@andrew.cmu.edu, moura@ece.cmu.edu, ph: (412)-268-6341, fax: (412)-268-3980)} 
}
\date{}
\maketitle
\begin{abstract}
We present \emph{telescoping} recursive representations for both continuous and discrete indexed noncausal Gauss-Markov random fields.  Our recursions start at the boundary (a hypersurface in $\R^d$, $d \ge 1$) and telescope inwards.  For example, for images, the telescoping representation reduce recursions from $d = 2$ to $d = 1$, i.e., to recursions on a single dimension.  Under appropriate conditions, the recursions for the random field are linear stochastic differential/difference equations driven by white noise, for which we derive recursive estimation algorithms, that extend standard algorithms, like the Kalman-Bucy filter and the Rauch-Tung-Striebel smoother, to noncausal Markov random fields.
\end{abstract}


\section{Introduction}

We consider the problem of deriving recursive representations for spatially distributed signals, such as temperature in materials, concentration of components in process control, intensity of images, density of a gas in a room, stress level of different locations in a structure, or pollutant concentration in a lake \cite{Van1983,Reu2005,PiciiCarli2008}.  These signals are often modeled using \emph{random fields}, which are random signals indexed over $\R^d$ or $\Z^d$, for $d \ge 2$.  For random \emph{processes}, which are indexed over $\R$, recursive algorithms are recovered by assuming causality.  In particular, for Markov random processes, the future states depend only on the present state given both the past and present states.  When modeling spatial distributions by random fields, it is more appropriate to assume noncausality as opposed to causality.  This leads to noncausal\footnote{When referring to causal or noncausal Markov random fields or processes, we really mean they admit recursive or nonrecursive representations.} Markov random fields (MRFs): the field inside a domain is independent of the field outside the domain given the field on (or near) the domain boundary.  The need for recursive algorithms for noncausal MRFs arises to reduce the increased computational complexity due to the noncausality and the multidimensionality of the index set. 
The assumption of noncausality presents problems in developing recursive algorithms such as the Kalman-Bucy filter for noncausal MRFs.
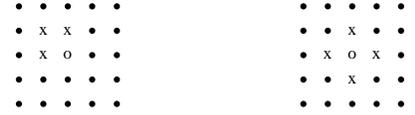
\begin{figure}
\begin{center}
\subfigure[Causal Structure]{
\scalebox{0.65}{
\begin{pspicture}[unit=0.5cm](-3,-1)(2,2)
\psdots(-2,2)(-1,2)(0,2)(1,2)(2,2)
(-2,1)(1,1)(2,1)(-2,0)(1,0)(2,0)
(-2,-1)(-1,-1)(0,-1)(1,-1)(2,-1)
(-2,-2)(-1,-2)(0,-2)(1,-2)(2,-2)
\rput[a](0,0){o} 
\rput[a](-1,0){x}
\rput[a](-1,1){x}
\rput[a](0,1){x}
\end{pspicture}
}
\label{fig:causal}
}
\subfigure[Noncausal Structure]{
\scalebox{0.65}{
\begin{pspicture}[unit=0.5cm](-3,-1)(2,2)
\psdots(-2,2)(-1,2)(0,2)(1,2)(2,2)
(-2,1)(1,1)(2,1)(-2,0)(2,0)
(-2,-1)(-1,-1)(-1,1)(1,-1)(2,-1)
(-2,-2)(-1,-2)(0,-2)(1,-2)(2,-2)
\rput[a](0,0){o} 
\rput[a](-1,0){x}
\rput[a](1,0){x}
\rput[a](0,1){x}
\rput[a](0,-1){x}
\end{pspicture}
\label{fig:noncausal}
}}
\caption{Causal and Noncausal models for random fields}
\vspace{-0.5cm}
\label{fig:folding}
\end{center}
\end{figure}

Instead, to derive recursive algorithms, many authors make causal approximations to random fields over $\R^2$ or $\Z^2$, see \cite{Abend1965, Habibi1972, Woods1977, Jain1977,Pickard1977, Marzetta1980, OgierWong1981, Goutsias1989}.  An example of a random field with causal structure is shown in Fig. \ref{fig:causal}.  It is assumed that the site indicated by `o' depends on the neighbors indicated by `x'.  Such fields do not capture fully the spatial dependence, as for example, when the field at a spatial location depends on its neighbors.  More appropriate representations are noncausal models, an example of which is the nearest neighbor model shown in Fig. \ref{fig:noncausal}.  In \cite{LevyAdamsWillsky1990}, the authors derive recursive estimation equations for \emph{nearest neighbor models} over $\Z^2$ by stacking two rows (or columns) at a time of the lattice into one vector and thus converting the two-dimensional (2-D) estimation problem into a one-dimensional (1-D) estimation problem with state of dimension $2n$, for an $n \times n$ lattice.  However, the algorithm in \cite{LevyAdamsWillsky1990} is restricted to nearest neighbor models with boundary conditions being local, i.e., they involve only neighboring points along the boundary.
In \cite{MouraBalram1992}, the authors derive a recursive representation for general noncausal Gauss-Markov random fields (GMRFs) over $\Z^2$ by stacking the field in each row (or column) and factoring the field covariance to get 1-D state-space models.  However, the models in 
\cite{MouraBalram1992} are only valid when the boundary conditions are assumed to be zero.  Further, since we can not stack columns or rows over a continuous index space, it is not clear how the methods of \cite{LevyAdamsWillsky1990} and \cite{MouraBalram1992} can be extended to derive recursive representations for noncausal GMRFs over $\R^d$ for $d \ge 2$.

For noncausal \emph{isotropic} GMRFs over $\R^2$, the authors in \cite{TewfikLevyWillsky1991} derived recursive representations, and subsequently recursive estimators, by transforming the 2-D problem into a countably infinite number of 1-D problems.  This transformation was possible because of the isotropy assumption since isotropic fields over $\R^2$, when expanded in a Fourier series in terms of the polar coordinate angle, the Fourier coefficient processes of different orders are uncorrelated \cite{TewfikLevyWillsky1991}.  In this way, the authors derived recursive representations for the Fourier coefficient process.  The recursions in \cite{TewfikLevyWillsky1991} are with respect to the radius when the field is represented in polar coordinate form.  The algorithm is an approximate recursive estimation algorithm since it requires solving a set of countably \emph{infinite} number of 1-D estimation problems \cite{TewfikLevyWillsky1991}.  For random fields with discrete indices, \emph{nonrecursive} approximate estimation algorithms can be found in the literature on estimation of graphical models, e.g., \cite{WainwrightJordan2008}.

In this paper, we present a telescoping recursive representation for general noncausal Gauss-Markov random fields defined on a closed continuous index set in $\R^d$, $d \ge 2$, or on a closed discrete index set in $\Z^d$, $d \ge 2$.  The telescoping recursions initiate at the boundary of the field and recurse inwards.  For example, in Fig.~\ref{fig:recursions}(a), for a GMRF defined on a unit disc, we derive telescoping representations that recurse radially inwards to the center of the field. 
For the same field, we derive an equivalent representation where the telescoping surfaces are not necessarily symmetric about the center of the disc, see Fig. \ref{fig:recursions}(b).  Further, the telescoping surfaces, under appropriate conditions, can be arbitrary as shown in Fig.~\ref{fig:recursions}(c).  In general, for a field indexed in $\R^d$, $d \ge 2$, the corresponding telescoping surfaces will be hypersurfaces in $\R^d$.
We parametrize the field using two parameters: $\lambda \in [0,1]$ and $\theta \in \Theta \subset \R^{d-1}$.  The parameter $\lambda$ indicates the position of the telescoping surface and the set $\Theta$ parameterizes the boundary of the index set.  For example, for the unit disc with recursions as in 
Fig.~\ref{fig:recursions}(a), the telescoping surfaces are circles, and we can use polar coordinates to parameterize the field: radius $\lambda$ and angle $\theta \in \Theta = [-\pi,\pi]$.  The telescoping surfaces are represented using a \emph{homotopy} from the boundary of the field to a point within the index set (which is not on the boundary).  The net effort for $d = 2$ is to represent the field by a recursion in $\lambda$, i.e., a single parameter (or dimension) rather than multiple dimensions.

The key idea in deriving the telescoping representation is to establish a notion of ``time" for Markov random fields.  We show that the parameter $\lambda$, which corresponds to the telescoping surface, acts as time.  In our telescoping representation, we define the state to be the field values at the telescoping surfaces.  The telescoping recursive representation we derive is a linear stochastic differential equation in the parameter $\lambda$ and is driven by Brownian motion.  For a certain class of \emph{homogeneous isotropic} GMRFs over $\R^2$, for which the covariance is a function of the Euclidean distance between points, we show that the driving noise is 2-D white Gaussian noise.  For the Whittle field \cite{Whittle1954} defined over a unit disc, we show that the driving noise is zero and the field is uniquely determined using the boundary conditions.

Using the telescoping recursive representation, we promptly recover recursive algorithms, such as the Kalman-Bucy filter \cite{KalmanBucy1961} and the Rauch-Tung-Striebel (RTS) smoother \cite{RauchTungStriebel1965}.  For the Kalman-Bucy filter, we sweep the observations over the telescoping surfaces starting at the boundary and recursing inwards.  For the smoother, we sweep the observations starting from the inside and recursing outwards.  Although, we use the RTS smoother in this paper, other known smoothing algorithms can be used as well, see \cite{KailathSayed,BadawiLP1979,FerrantePicci2000}.

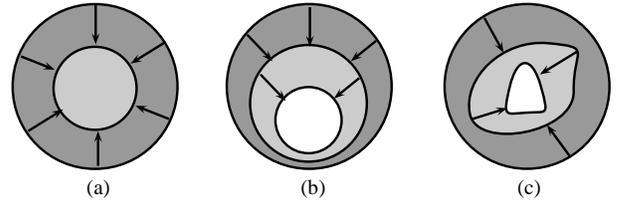
\begin{figure}
\begin{center}
\scalebox{0.8} 
{
\begin{pspicture}(0,-1.39)(9.96,1.39)
\definecolor{color2918b}{rgb}{0.6,0.6,0.6}
\definecolor{color66b}{rgb}{0.8,0.8,0.8}
\rput{-270.0}(8.57,-8.57){\pscircle[linewidth=0.04,dimen=outer,fillstyle=solid,fillcolor=color2918b](8.57,0.0){1.39}}
\rput{-270.0}(4.95,-4.95){\pscircle[linewidth=0.04,dimen=outer,fillstyle=solid,fillcolor=color2918b](4.95,0.0){1.39}}
\pscircle[linewidth=0.04,dimen=outer,fillstyle=solid,fillcolor=color2918b](1.39,0.0){1.39}
\pscircle[linewidth=0.04,dimen=outer,fillstyle=solid,fillcolor=color66b](1.39,-0.03){0.71}
\psline[linewidth=0.04cm,arrowsize=0.05291667cm 2.0,arrowlength=1.4,arrowinset=0.4]{->}(1.4,1.36)(1.4,0.7)
\psline[linewidth=0.04cm,arrowsize=0.05291667cm 2.0,arrowlength=1.4,arrowinset=0.4]{->}(2.62,-0.54)(2.06,-0.32)
\psline[linewidth=0.04cm,arrowsize=0.05291667cm 2.0,arrowlength=1.4,arrowinset=0.4]{->}(0.28,-0.74)(0.86,-0.38)
\psline[linewidth=0.04cm,arrowsize=0.05291667cm 2.0,arrowlength=1.4,arrowinset=0.4]{->}(1.44,-1.32)(1.44,-0.74)
\psline[linewidth=0.04cm,arrowsize=0.05291667cm 2.0,arrowlength=1.4,arrowinset=0.4]{<-}(0.72,0.28)(0.16,0.5)
\psline[linewidth=0.04cm,arrowsize=0.05291667cm 2.0,arrowlength=1.4,arrowinset=0.4]{<-}(1.98,0.38)(2.56,0.74)
\rput{-270.0}(4.66,-5.22){\pscircle[linewidth=0.04,dimen=outer,fillstyle=solid,fillcolor=color66b](4.94,-0.28){0.99}}
\rput{-270.0}(4.38,-5.5){\pscircle[linewidth=0.04,dimen=outer,fillstyle=solid](4.94,-0.56){0.57}}
\psline[linewidth=0.04cm,arrowsize=0.05291667cm 2.0,arrowlength=1.4,arrowinset=0.4]{->}(4.96,1.32)(4.96,0.66)
\psline[linewidth=0.04cm,arrowsize=0.05291667cm 2.0,arrowlength=1.4,arrowinset=0.4]{<-}(5.64,0.42)(6.06,0.76)
\psline[linewidth=0.04cm,arrowsize=0.05291667cm 2.0,arrowlength=1.4,arrowinset=0.4]{<-}(4.34,0.42)(3.92,0.86)
\psline[linewidth=0.04cm,arrowsize=0.05291667cm 2.0,arrowlength=1.4,arrowinset=0.4]{<-}(5.36,-0.2)(5.78,0.14)
\psline[linewidth=0.04cm,arrowsize=0.05291667cm 2.0,arrowlength=1.4,arrowinset=0.4]{<-}(4.56,-0.24)(4.14,0.2)
\psbezier[linewidth=0.04,fillstyle=solid,fillcolor=color66b](7.62,-0.46)(7.745445,-1.06)(9.357111,-0.78)(9.343739,-0.08)(9.330366,0.62)(9.68,0.64846694)(8.921322,0.72)(8.162643,0.79153305)(7.4945555,0.14)(7.62,-0.46)
\psbezier[linewidth=0.04,fillstyle=solid](8.280875,-0.0055799484)(8.406465,0.52)(8.655096,0.5196132)(8.797548,0.0057087666)(8.94,-0.50819564)(8.90061,-0.4120821)(8.556434,-0.43455112)(8.212258,-0.45702016)(8.155286,-0.5311599)(8.280875,-0.0055799484)
\psline[linewidth=0.04cm,arrowsize=0.05291667cm 2.0,arrowlength=1.4,arrowinset=0.4]{->}(7.86,1.14)(8.18,0.56)
\psline[linewidth=0.04cm,arrowsize=0.05291667cm 2.0,arrowlength=1.4,arrowinset=0.4]{->}(9.28,-1.16)(8.9,-0.64)
\psline[linewidth=0.04cm,arrowsize=0.05291667cm 2.0,arrowlength=1.4,arrowinset=0.4]{->}(9.4,0.58)(8.78,0.2)
\psline[linewidth=0.04cm,arrowsize=0.05291667cm 2.0,arrowlength=1.4,arrowinset=0.4]{->}(7.66,-0.54)(8.24,-0.36)
\usefont{T1}{ptm}{m}{n}
\rput(8.605156,-1.7){(c)}
\usefont{T1}{ptm}{m}{n}
\rput(5.0151563,-1.7){(b)}
\usefont{T1}{ptm}{m}{n}
\rput(1.4451562,-1.7){(a)}
\end{pspicture} 
}
\end{center}
\caption{Different kinds of telescoping recursions for a GMRF defined on a disc.}
\label{fig:recursions}
\end{figure}

We derive the telescoping representation in an abstract setting over index sets in $\R^d$, $d \ge 2$.  We can easily specialize this to index sets over $\Z^d$, $d \ge 2$.  We show an example of this for GMRFs defined over a lattice.  We see that, unlike the continuous index case that admits many equivalent telescoping recursions, the telescoping recursion for discrete index GMRFs is unique.


The organization of the paper is as follows.  Section \ref{sec:gmrf} reviews the theory of GMRFs.  Section \ref{sec:telescoping_cont} introduces the telescoping representation for GMRFs indexed on a unit disc.  Section \ref{sec:telescoping_cont_arbit} generalizes the telescoping representations to arbitrary domains.  Section \ref{sec:rec_estimation_gmrf} derives recursive estimation algorithms using the telescoping representation.  Section \ref{sec:tel_rep_gmrf} derives telescoping recursions for GMRFs with discrete indices.  Section \ref{sec:summary} summarizes the paper.

\section{Gauss-Markov Random Fields}
\label{sec:gmrf}

\subsection{Continuous Indices}

For a random process $x(t)$, $t \in \R$, the notion of Markovianity corresponds to the assumption that the past $\{x(s): s<t\}$, and the future $\{x(s): s>t\}$ are conditionally independent given the present $x(s)$.  Higher order Markov processes can be considered when the past is independent of the future given the present and information near the present.  The extension of this definition to random fields, \emph{i.e.}, a random process indexed over $\R^d$ for $d\ge 2$, was introduced in \cite{Levy1956}.  Specifically, a random field $x(t)$, $t \in T \subset \R^d$, is Markov if for any smooth surface $\partial G$ separating $T$ into complementary domains, the field inside is independent of the field outside conditioned on the field on (and near) $\partial G$.  To capture this definition in a mathematically precise way, we use the notation introduced in \cite{McKean1963}.  On the probability space $(\Omega, {\cal F}, {\cal P})$, let\footnote{For ease in notation, we assume $x(t) \in \R$, however our results remain valid for $x(t) \in \R^n$, when $n \ge 2$.} $x(t) \in \R$  be a zero mean random field for $t \in T \subset \R^d$, where $d \ge 2$ and let $\partial T \subset T$ be the smooth boundary of $T$.  For any set $A \subset T$, denote $x(A)$ as
\begin{equation}
x(A) = \{x(t):t \in A\} \,. 
\end{equation}
Let $G_{-} \subset T$ be an open set with smooth boundary $\partial G$ and let $G_{+}$ be the complement of $G_{-} \cup \partial G$ in $T$.  Together, $G_{-}$ and $G_{+}$ are called \emph{complementary sets}.  Fig. \ref{fig:notation}(a) shows an example of the sets $G_{-}$, $G_{+}$, and $\partial G$ on a domain $T \subset \R^2$.  \
\begin{figure}
\centering
\scalebox{0.8} 
{
\begin{pspicture}(0,-1.7984375)(10.14,1.7984375)
\psbezier[linewidth=0.04](0.21241651,0.45999998)(0.42483303,1.4371792)(1.7434434,1.5553567)(2.6924167,1.24)(3.6413898,0.9246431)(3.9810684,0.831175)(3.6724164,-0.12000002)(3.3637645,-1.0711751)(2.5910704,-1.0055151)(1.6124164,-0.8)(0.63376254,-0.5944849)(0.0,-0.51717925)(0.21241651,0.45999998)
\psbezier[linewidth=0.04](1.0746561,0.17423652)(1.1656368,-0.5365875)(1.8273737,-0.64000005)(2.289895,-0.0981443)(2.7524164,0.44371143)(2.6667325,0.3643695)(1.9924164,0.65999997)(1.3181002,0.9556305)(0.98367554,0.8850605)(1.0746561,0.17423652)
\usefont{T1}{ptm}{m}{n}
\rput(1.628979,0.30999997){$G_{-}$}
\usefont{T1}{ptm}{m}{n}
\rput(3.0153852,-0.07){$G_{+}$}
\psline[linewidth=0.04cm,arrowsize=0.05291667cm 2.5,arrowlength=1.5,arrowinset=0.4]{->}(4.572417,0.98)(3.7724166,0.58)
\usefont{T1}{ptm}{m}{n}
\rput(4.983354,1.13){$\partial T$}
\psline[linewidth=0.04cm,arrowsize=0.05291667cm 2.5,arrowlength=1.5,arrowinset=0.4]{->}(3.0524166,1.5)(2.2724166,0.56)
\usefont{T1}{ptm}{m}{n}
\rput(3.2853854,1.61){$\partial G$}
\psline[linewidth=0.04cm,dotsize=0.07055555cm 2.0]{*-*}(6.22,0.3184375)(10.12,0.3184375)
\psline[linewidth=0.04cm](7.76,0.4584375)(7.76,0.1784375)
\usefont{T1}{ptm}{m}{n}
\rput(7.685385,-0.69){$\partial G$}
\usefont{T1}{ptm}{m}{n}
\rput(6.928979,0.71){$G_{-}$}
\usefont{T1}{ptm}{m}{n}
\rput(8.975385,0.71){$G_{+}$}
\psline[linewidth=0.04cm,arrowsize=0.05291667cm 2.0,arrowlength=1.4,arrowinset=0.4]{->}(7.64,-0.4015625)(7.74,0.1184375)
\usefont{T1}{ptm}{m}{n}
\rput(1.9251562,-1.5715625){(a)}
\usefont{T1}{ptm}{m}{n}
\rput(8.095157,-1.5515625){(b)}
\end{pspicture} 
}\caption{(a) An example of complementary sets on a random field defined on $T$ with boundary $\partial T$. (b) Corresponding notion of complementary sets for a random process.}
\label{fig:notation}
\end{figure}
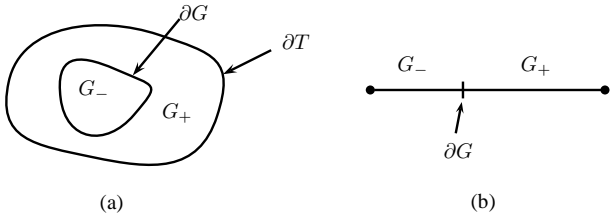
For $\epsilon > 0$, define the set of points from the boundary $\partial G$ at a distance less than $\epsilon$
\begin{equation}
\partial G_{\epsilon} = \{t \in T : d(t,\partial G) < \epsilon \} \,,
\end{equation}
where $d(t,\partial G)$ is the distance of a point $t \in T$ to the set of points $\partial G$.  On $G_{\pm}$ and $\partial G$, define the sets 
\begin{align}
\Sigma_x(G_{\pm}) &= \sigma(x(G_{\pm})) \\
\overline{\Sigma}_x(\partial G) &= \bigcap_{\epsilon > 0} 
\sigma(x(\partial G_{\epsilon})) \,, \label{eq:germ_sigma}
\end{align}
where $\sigma(A)$ stands for the $\sigma$-algebra generated by the set $A$.  If $x(t)$ is a \emph{Markov random field}, the conditional expectation of $x(s)$, $s \notin G_{-}$, given $\Sigma_x(G_{-})$ is the conditional expectation of $x(s)$ given $\overline{\Sigma}_x(\partial G)$, i.e., \cite{McKean1963, Ldpitt1971}
\begin{equation}
E[x(s) | \Sigma_x(G_{-})] = E[x(s) | \overline{\Sigma}_x(\partial G)] \,, \quad s \notin G_{-}  \,. \label{eq:def_mrf}
\end{equation}
Equation (\ref{eq:def_mrf}) also holds for Markov random processes for complementary sets defined as in Fig. \ref{fig:notation}(b).  In the context of Markov processes, the set $G_{-}$ in 
Fig.~\ref{fig:notation}(b) is called the ``past", $G_{+}$ is called the ``future", and $\partial G$ is called the ``present".  The equivalent notions of past, present, and future for random fields is clear from the definition of $G_{-}$, $G_{+}$, and $\partial G$ in Fig.~\ref{fig:notation}(a).

In this paper, we assume $x(t)$ is zero mean Gaussian, giving us a Gauss-Markov random field (GMRF), so the conditional expectation in (\ref{eq:def_mrf}) becomes a linear projection.  Following \cite{Ldpitt1971}, the key assumptions we make throughout the paper are as follows.
\begin{enumerate}[{A}1.]
\item We assume the index set $T \subset \R^d$ is a connected\footnote{A set is connected if it can not be divided into disjoint nonempty closed set.} open set with smooth boundary $\partial T$.
\item The zero mean GMRF $x(t) \in L^2(\Omega, {\cal F}, {\cal P})$, which means that $x(t)$ has finite energy.
\item The covariance of $x(t)$ is $R(t,s)$, where $t,s \in T \subset \R^d$.  The function space of $R(t,s)$ is associated with the uniformly strongly elliptic inner product
\begin{align}
<u,v> &= \left\langle D^{\alpha} u, a_{\alpha,\beta} D^{\beta}v\right\rangle_T \\
&= \sum_{|\alpha|\le m, |\beta|\le m}\int_T D^{\alpha} u(s) a_{\alpha,\beta}(s) D^{\beta}v(s) ds \,,
\label{eq:bilinear_form}
\end{align}
where $a_{\alpha,\beta}$ are bounded, continuous, and infinitely differentiable, $\alpha = [\alpha_1,\cdots,\alpha_d]$ is a multi-index of order $|\alpha| = \alpha_1 + \cdots + \alpha_d$ and the operator $D^{\alpha}$ is the partial derivative operator 
\begin{equation}
D^{\alpha} = D_1^{\alpha_1} \cdots D_d^{\alpha_d} \,,
\end{equation}
where $D_i^{\alpha_i} = \partial^{\alpha_i} / \partial t_i^{\alpha_i}$ for $t = [t_1,\ldots,t_d]$.
\item Since the inner product in (\ref{eq:bilinear_form}) is uniformly strongly elliptic, it follows as a consequence of A3 that $R(t,s)$ is jointly continuous, and thus $x(t)$ can be modified to have continuous sample paths.  We assume that this modification is done, so the GMRF $x(t)$ has continuous sample paths.
\end{enumerate}
Under Assumptions A1-A4, we now review results on GMRFs we use in the paper.

\noindent
\textbf{Weak normal derivatives:}
Let $\partial G$ be a boundary separating complementary sets $G_{-}$ and $G_{+}$.  Whenever we refer to normal derivatives, they are to be interpreted in the following weak sense: For every smooth $f(t)$,
\begin{align}
y(s) &= \frac{\partial}{\partial n} x(s) \nonumber \\
\Rightarrow \int_{\partial G} f(s) y(s) dl
&= \lim_{h \rightarrow 0} \frac{\partial }{\partial h}\int_{\partial G} f(s) x(s+h\dot{s}) dl \,, \label{eq:normalized_deriv}
\end{align}
where $dl$ is the surface measure on $\partial G$ and $\dot{s}$ is the unit vector normal to $\partial G$ at the point $s$.

\noindent
\textbf{GMRFs with order $m$:}  Throughout the paper, unless mentioned otherwise, we assume that the GMRF has order $m$, which can have multiple different equivalent interpretation: (i) the GMRF $x(t)$ has $m-1$ normal derivatives, defined in the weak sense, for each point $t \in \partial G$ for all possible surfaces $\partial G$, (ii) the $\sigma$-algebra $\overline{\Sigma}_x(\partial G)$ in (\ref{eq:germ_sigma}), called the germ $\sigma$-algebra, contains information about $m-1$ normal derivatives of the field on the boundary $\partial G$ \cite{Ldpitt1971}, or (iii) there exists a symmetric and positive strongly elliptic differential operator ${\cal L}_t$ with order $2m$ such that \cite{MouraGoswami1997}
\begin{equation}
{\cal L}_t R(t,s) = \delta(t-s) \label{eq:lr=d}\,,
\end{equation}
where the differential operator has the form,
\begin{equation}
{\cal L}_t u(t) = \sum_{|\alpha|,|\beta| \le m}
(-1)^{|\alpha|} D^{\alpha}[ a_{\alpha,\beta}(t) D^{\beta}(u(t))] \,.
\label{eq:l_t}
\end{equation}

\noindent
\textbf{Prediction:}  The following theorem, proved in \cite{Ldpitt1971}, gives us a closed form expression for the conditional expectation in (\ref{eq:def_mrf}).

\begin{theorem}[$\!\!$\cite{Ldpitt1971}]
\label{thm:prediction}
Let $x(t)$, $t \in T \subset \R^d$, be a zero mean GMRF of order $m$ and covariance $R(t,s)$.  Consider complementary sets $G_{-}$ and $G_{+}$ with common boundary $\partial G$.  For $s \notin G_{-}$, the conditional expectation of $x(s)$ given $\Sigma_x(G_{-})$ is 
\begin{equation}
E[x(s) | \Sigma_x(G_{-})] 
= \sum_{j=0}^{m-1} \int_{\partial G}  b_{j}(s,r) \frac{\partial^j}{\partial n^j} x(r) d l \,, \label{eq:prediction_2}
\end{equation}
where $\partial^j/\partial n^j$ is the normal derivative, defined in (\ref{eq:normalized_deriv}), $dl$ is a surface measure on the boundary $\partial G$, and the functions $b_j(s,r)$, $s \notin G_{-}$ and $r \in \partial G$, are smooth.
\end{theorem}
\begin{IEEEproof}
A detailed proof of Theorem \ref{thm:prediction} can be found in \cite{Ldpitt1971}, where the result is proved for the case when $s \in G_{+}$.  To include the case when $s \in \partial G$, we use the fact that $R(t,s)$ is jointly continuous (consequence of A3) and the uniform integrability of the Gaussian measure (see \cite{Anandkumar2009ISIT}).
\end{IEEEproof}

Theorem~\ref{thm:prediction} says that for each point outside $G_{-}$, the conditional expectation given all the points in $G_{-}$ depends only the field defined on or near the boundary.
This is not surprising since, as stated before, $E[x(s)|\Sigma_x(G_{-})] = E[x(s) | \overline{\Sigma}_x(\partial G)]$, and we mentioned before that $\overline{\Sigma}_x(\partial G)$ has information about the $m-1$ normal derivatives of $x(t)$ on the surface $\partial G$. Appendix \ref{app:example_gmrf} shows how the smooth functions $b_j(s,r)$ can be computed and outlines an example of the computations in the context of a Gauss-Markov process.  In general, Theorem~\ref{thm:prediction} extends the notion of a Gauss-Markov \emph{process} of order $m$ (or an autoregressive process of order $m$) to random fields.

A simple consequence of Theorem \ref{thm:prediction} is that we get the following characterization for the covariance of a GMRF of order $m$.

\begin{theorem}
\label{thm:property_rts}
If $t \in G_{-}$ and $s \notin G_{-}$, the covariance $R(t,s)$ can be written as,
\begin{equation}
R(s,t) =  \sum_{j=0}^{m-1} \int_{\partial G}  b_{j}(s,r) \frac{\partial^j}{\partial n^j} R(r,t) d l \,, \label{eq:property_rts}
\end{equation}
where the normal derivative in (\ref{eq:property_rts}) is with respect to the variable $r$.
\end{theorem}
\begin{IEEEproof}
Since $x(s) - E[x(s)|\Sigma_x(G_{-})] \perp x(t)$ for $t \in G_{-}$, using (\ref{eq:prediction_2}), we can easily establish (\ref{eq:property_rts}).
\end{IEEEproof}
Theorem~\ref{thm:property_rts} says that the covariance $R(s,t)$ of a GMRF can be written in terms of the covariance of the field defined on a boundary dividing $s$ and $t$.  Both Theorems \ref{thm:prediction} and \ref{thm:property_rts} will be used in deriving the telescoping recursive representation.

\subsection{Discrete Indices}
\label{sec:discrete_indices}

Discrete index Markov random fields, also known as undirected graphical models, are characterized by interactions of an index point with its neighbors.  In this paper, we only consider GMRFs defined on a lattice $T_0 = [0, N+1] \times [0,M+1]$.  An index $(i,j) \in T_0$ will be called a \emph{node}.  If two nodes are neighbors of each other, we represent this relationship by connecting them with an edge.  A \emph{path} is the set of distinct nodes visited when hopping from node $(i_1,j_1)$ to a node $(i_2,j_2)$ where the hops are only along edges.  A subset of sites $C$ separates two sites $(i_1,j_1) \notin C$ and $(i_2,j_2) \notin C$ if every path from $(i_1,j_1)$ to $(i_2,j_2)$ contains at least one node in C.  Two disjoint sets $ A,B \subset T \backslash C$ are separated by $C$ if every pair of sites, one in $A$ and the other in $B$, are separated by $C$.

We denote the discrete index random field by $x(i,j) \in \R$.  Let ${\cal N}$ denote the neighborhood structure for the random field, then $x(i,j)$ is a GMRF if $x(i,j)$ is independent of $x\left(T_0 \backslash \{{\cal N}\cup(i,j)\}\right)$ given $x({\cal N})$ for $(i,j) \in T_0 \backslash \partial T_0$, where $\partial T_0$ denotes the boundary nodes of $T_0$.  An equivalent way to define GMRFs is using the global Markov property:

\begin{theorem}[Global Markov property\cite{SpeedKiiveri1986}]
\label{thm:global_markov_property}
For a GMRF $x(i,j)$ for $(i,j) \in T_0 = [0,N+1]\times [0,M+1]$, for all disjoint sets $A$, $B$, and $C$ in $T_0$, where $A$ and $B$ are non-empty and $C$ separates $A$ and $B$, $x(A)$ is independent of $x(B)$ given $x(C)$.
\end{theorem}

For ease in notation and simplicity, we only consider second order neighborhoods denoted by the set ${\cal N}_2$ such that for node $(0,0)$:
\begin{align}
{\cal N}_2 = \{(-1,0),(1,0),(0,-1),(0,1),\nonumber \\
(1,\pm 1),(\pm 1,1),(1,1),(-1,-1) \} \,.
\end{align}
Examples of higher order neighborhood structures are shown in Fig. \ref{fig:neighbor}.  A nonrecursive representation, derived in \cite{Woods1972}, for $x(i,j)$ is given as follows:
\begin{align}
\alpha_{i,j} x(i,j) &= \sum_{(k,l) \in {\cal N}_2} \beta_{ij}^{k,l} x(i-k,j-l) \label{eq:mmse}\\
&\hspace{1.5cm}+ v(i,j) \,, (i,j) \in T_0 \backslash \partial T_0 \nonumber \,,
\end{align}
where $v_{i,j}$ is locally correlated noise such that
\begin{align*}
E[v(i,j) x(k,l)] &= \delta(i-k) \delta(j-l) \\
E[v(i,j) v(i,j)] &= \left\{\begin{array}{cc}
0 & (k-i,l-j) \notin {\cal N}_2 \\
\alpha_{i,j} & k = i, j = l \\
-\beta_{ij}^{k-i,l-j} & (k-i,l-j) \in {\cal N}_2
\end{array}\right. \,.
\end{align*}
Since $E[v(i,j) v(k,l)] = E[v(k,l) v(i,j)]$, we have
\begin{equation} \beta_{ij}^{k-i,l-j} = \beta_{kl}^{i-k,k-l} \,.
\label{eq:symmetric_beta}
\end{equation}
The boundary conditions in (\ref{eq:mmse}) are assumed to be Dirichlet such that $x(\partial T_0)$ is Gaussian with zero mean and known covariance.  

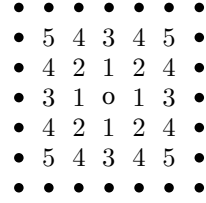
\begin{figure}
\begin{center}
\begin{pspicture}[unit=0.4cm](-3,-1.5)(3,1.5)
\psdots(-3,3)(-2,3)(-1,3)(0,3)(1,3)(2,3)(3,3)
(-3,2)(3,2)
(-3,1)(3,1)
(-3,0)(3,0)
(-3,-1)(3,-1)
(-3,-2)(3,-2)
(-3,-3)(-2,-3)(-1,-3)(0,-3)(1,-3)(2,-3)(3,-3)
\rput[a](0,0){o} 
\rput[a](0,1){$1$} 
\rput[a](-1,0){$1$} 
\rput[a](0,-1){$1$} 
\rput[a](1,0){$1$} 
\rput[a](1,1){$2$}
\rput[a](1,-1){$2$}
\rput[a](-1,-1){$2$}
\rput[a](-1,1){$2$}
\rput[a](2,0){$3$}
\rput[a](-2,0){$3$}
\rput[a](0,2){$3$}
\rput[a](0,-2){$3$}
\rput[a](2,1){$4$}
\rput[a](2,-1){$4$}
\rput[a](-2,-1){$4$}
\rput[a](-2,1){$4$}
\rput[a](-1,2){$4$}
\rput[a](-1,-2){$4$}
\rput[a](1,-2){$4$}
\rput[a](1,2){$4$}
\rput[a](2,2){$5$}
\rput[a](2,-2){$5$}
\rput[a](-2,2){$5$}
\rput[a](-2,-2){$5$}
\end{pspicture}
\caption{Neighborhood structure from order 1 to 5.}
\label{fig:neighbor}
\end{center}
\end{figure}

\section{Telescoping Representation: GMRFs on a Unit Disc}
\label{sec:telescoping_cont}

In this Section, we present the telescoping recursive representation for GMRFs indexed over a domain $T \subset \R^2$, which is assumed to be a unit disc centered at the origin.  The generalization to arbitrary domains is presented in Section \ref{sec:telescoping_cont_arbit}.  To parametrize the GMRF, say $x(t)$ for $t \in T$, we use polar coordinates such that $x_{\lambda}(\theta)$ is defined to be the point
\begin{equation}
x_{\lambda}(\theta) = x((1-\lambda)\cos \theta,(1-\lambda)\sin \theta) \,,
\label{eq:polar_coord}
\end{equation}
where $(\lambda,\theta) \in [0,1] \times [-\pi,\pi]$.  Thus, $\{x_{0}(\theta): \theta \in [-\pi,\pi]\}$ corresponds to the field defined on the boundary of the unit disc, denoted as $\partial T$.  Let $\partial T^{\lambda}$ denote the set of points in $T$ at a distance $1-\lambda$ from the center of the field.  We call $\partial T^{\lambda}$ a \emph{telescoping surface} since the telescoping representations we derive recurse these surfaces.  The notations introduced so far are shown in Fig. \ref{fig:notations_unit_disc}.

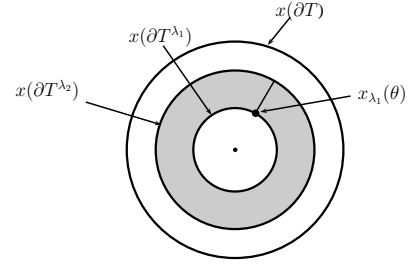
\begin{figure}
\begin{center}
\scalebox{0.5} 
{
\begin{pspicture}(0,-3.3973436)(14.284219,3.4373438)
\definecolor{color396b}{rgb}{0.8,0.8,0.8}
\pscircle[linewidth=0.06,dimen=outer,fillstyle=solid](7.888125,-0.48734376){2.91}
\pscircle[linewidth=0.06,dimen=outer,fillstyle=solid,fillcolor=color396b](7.888125,-0.48734376){2.14}
\pscircle[linewidth=0.06,dimen=outer,fillstyle=solid](7.888125,-0.48734376){1.14}
\psdots[dotsize=0.2](8.438125,0.48265624)
\psdots[dotsize=0.12](7.888125,-0.48734376)
\psline[linewidth=0.04cm,arrowsize=0.05291667cm 2.0,arrowlength=1.4,arrowinset=0.4]{->}(9.358125,3.0426562)(8.758125,2.3026562)
\usefont{T1}{ptm}{m}{n}
\rput(9.605157,3.1926563){\Large $x(\partial T$)}
\psline[linewidth=0.04cm,arrowsize=0.05291667cm 2.0,arrowlength=1.4,arrowinset=0.4]{->}(3.9181252,0.78265625)(5.925469,0.17734376)
\psline[linewidth=0.04cm,arrowsize=0.05291667cm 2.0,arrowlength=1.4,arrowinset=0.4]{->}(5.978125,2.2826562)(7.2581253,0.44265625)
\psline[linewidth=0.04cm,arrowsize=0.05291667cm 2.0,arrowlength=1.4,arrowinset=0.4]{->}(10.958125,0.9426563)(8.5581255,0.52265626)
\usefont{T1}{ptm}{m}{n}
\rput(2.9,1.1){\Large $x(\partial T^{\lambda_2})$}
\usefont{T1}{ptm}{m}{n}
\rput(5.905469,2.5326562){\Large $x(\partial T^{\lambda_1})$}
\usefont{T1}{ptm}{m}{n}
\rput(11.803594,0.99265623){\Large $x_{\lambda_1}(\theta)$}
\psline[linewidth=0.04cm](8.445469,0.51734376)(8.925468,1.3373437)
\end{pspicture} 
}
\end{center}
\caption{A random field defined on a unit disc.  The boundary of the field, i.e., the field values defined on the circle with radius $1$ is denoted by $x(\partial T)$.  The field values at a distance of $1-\lambda$ from the center of the field are given by $x(\partial T^{\lambda})$.  Each point is characterized in polar coordinates as $x_{\lambda}(\theta)$, where $1-\lambda$ is the distance to the center and $\theta$ denotes the angle.}
\label{fig:notations_unit_disc}
\end{figure}

\subsection{Main Theorem}

Before deriving our main theorem regarding the telescoping representation, we first define some notation.  Let $x(t) \in \R$ be a zero mean GMRF defined on a unit disc $T \subset \R^2$ parametrized as $x_{\lambda}(\theta)$, defined in (\ref{eq:polar_coord}).  Let $\Theta = [-\pi,\pi]$ and denote the covariance between $x_{\lambda_1}(\theta_1)$ and $x_{\lambda_2}(\theta_2)$ by $R_{\lambda_1,\lambda_2}(\theta_1,\theta_2)$ such that
\begin{equation}
R_{\lambda_1,\lambda_2}(\theta_1,\theta_2) = E[x_{\lambda_1}(\theta_1) x_{\lambda_2}(\theta_2)] \,.
\end{equation}
Define $C_{\lambda}(\theta_1,\theta_2)$ and $B_{\lambda}(\theta)$ as
\begin{align}
C_{\lambda}(\theta_1,\theta_2) &= \lim_{\mu \rightarrow \lambda^-} \frac{\partial}{\partial \mu} R_{\mu,\lambda}(\theta_1,\theta_2) - \lim_{\mu \rightarrow \lambda^+} \frac{\partial}{\partial \mu} R_{\mu,\lambda}(\theta_1,\theta_2) \label{eq:c_t_unit_disc} \\
B_{\lambda}(\theta) &= 
\left\{
\begin{array}{cc}
\sqrt{C_{\lambda}(\theta,\theta)} & C_{\lambda}(\theta,\theta) \ne 0 \\
K & C_{\lambda}(\theta,\theta) = 0
\end{array}
\right. \,, \label{eq:b_l_unit_disc}
\end{align}
where $K$ is any non-zero constant.  We will see in (\ref{eq:c_l_theta_theta}) that $C_{\lambda}(\theta,\theta)$ is the variance of a random variable and hence it is non-negative.  
Define $F_{\theta}$ as the integral transform
\begin{equation}
F_{\theta} [x_{\lambda}(\theta)] = \sum_{j=0}^{m-1} \int_{\Theta} \lim_{\mu \rightarrow \lambda^+} \frac{\partial}{\partial \mu}
b_j((\mu,\theta),(\lambda,\alpha))
\frac{\partial^j}{\partial n^j} x_{\lambda}(\alpha) d\alpha \label{eq:f_theta_unit_disc}\,,
\end{equation}
where $b_j((\mu,\theta),(\lambda,\alpha))$ is defined in (\ref{eq:prediction_2}) and the index $(\mu,\theta)$ in polar coordinates corresponds to the point $((1-\mu)\cos\theta,(1-\mu)\sin \theta)$ in Cartesian coordinates.  We see that $F_{\theta}[x_{\lambda}(\theta)]$ operates on the surface $\partial T^{\lambda}$ such that it is a linear combination of all normal derivatives of $x_{\lambda}({\theta})$ up to order $m-1$.
The normal derivative in (\ref{eq:f_theta_unit_disc}) is interpreted in the weak sense as defined in (\ref{eq:normalized_deriv}).  We now state the main theorem of the paper.


\begin{theorem}[\textbf{Telescoping Recursive Representation}]
\label{thm:telescoping_rec_gmrf_unit_disc}
For the GMRF parametrized as $x_{\lambda}(\theta)$, defined in (\ref{eq:polar_coord}), we have the following stochastic differential equation

\smallskip

\noindent
Telescoping Representation:
\begin{equation}
dx_{\lambda}(\theta) = F_{\theta} [x_{\lambda}(\theta)] d\lambda + B_{\lambda}(\theta) dw_{\lambda}(\theta) \label{eq:telescoping_unit_disc} \,,
\end{equation}
where $d x_{\lambda}(\theta) = x_{\lambda+d\lambda}(\theta) - x_{\lambda}(\theta)$ for $d\lambda$ small, $F_{\theta}$ is defined in (\ref{eq:f_theta_unit_disc}), $B_{\lambda}(\theta)$ is defined in (\ref{eq:b_l_unit_disc}), and
$w_{\lambda}(\theta)$ has the following properties:
\begin{enumerate}[i)]
\item The driving noise $w_{\lambda}(\theta)$ is zero mean Gaussian, almost surely continuous in $\lambda$, and independent of $x(\partial T)$ (the field on the boundary).
\item For all $\theta \in \Theta$, $w_{0}(\theta) = 0$.
\item For $0 \le \lambda_1 \le \lambda_1' \le \lambda_2 \le \lambda_2'$ and $\theta_1,\theta_2 \in \Theta$, $w_{\lambda_1'}(\theta_1) - w_{\lambda_1}(\theta_1)$ and $w_{\lambda_2'}(\theta_2) - w_{\lambda_2}(\theta_2)$ are independent random variables.
\item For $\theta \in \Theta$, we have
\begin{equation}
E[w_{\lambda}(\theta_1) w_{\lambda}(\theta_2)] = \int_0^{\lambda}
\frac{C_u(\theta_1,\theta_2)}{B_u(\theta_1) B_u(\theta_2)} du 
\label{eq:w_cov_l_1} \,.
\end{equation}

\item Assuming the set $\{u \in [0,1]: C_u(\theta,\theta) = 0\}$ has measure zero for each $\theta \in \Theta$, for $\lambda_1 > \lambda_2$ and $\theta_1,\theta_2 \in \Theta$, the random variable $w_{\lambda_1}(\theta) - w_{\lambda_2}(\theta)$ is Gaussian with mean zero and covariance
\begin{align}
E\left[\left(w_{\lambda_1}(\theta_1)-w_{\lambda_2}(\theta_2) \right)^2\right] & \nonumber \\
&\hspace{-2.5cm}= \lambda_1 + \lambda_2 - 2 \int_0^{\lambda_2} \frac{C_u(\theta_1,\theta_2)}{B_u(\theta_1) B_u(\theta_2)} du \label{eq:w_cov_l_2} \,.
\end{align}
\end{enumerate}
\end{theorem}
\begin{IEEEproof}
See Appendix A.
\end{IEEEproof}

Theorem~\ref{thm:telescoping_rec_gmrf_unit_disc} says that $x_{\lambda+d\lambda}(\theta)$, where $d\lambda$ is small, can be computed using the random field defined on the telescoping surface $\partial T^{\lambda}$ and some random noise.  The dependence on the telescoping surface follows from Theorem~\ref{thm:prediction}.
The main contribution in Theorem~\ref{thm:telescoping_rec_gmrf_unit_disc} is to explicitly compute properties of the driving noise $w_{\lambda}(\theta)$.  We now discuss the telescoping representation and highlight its various properties.

\subsubsection{\textbf{Driving noise $w_{\lambda}(\theta)$}}
The properties of the driving noise $w_{\lambda}(\theta)$ in (\ref{eq:telescoping_unit_disc}) lead to the following theorem.

\begin{theorem}[Driving noise $w_{\lambda}(\theta)$]
\label{thm:w_l_noise}
For the collection of random variables 
\[\{w_{\lambda}(\theta) : (\lambda,\theta)\in [0,1]\times\Theta \}\] 
defined in (\ref{eq:telescoping_unit_disc}), for each fixed $\theta \in \Theta$, $w_{\lambda}(\theta)$ is a standard Brownian motion when the set $\{u \in [0,1] : C_u(\theta,\theta) = 0\}$ has measure zero for each $\theta \in \Theta$.
\end{theorem}
\begin{IEEEproof}
For fixed $\theta \in \Theta$, to show $w_{\lambda}(\theta)$ is Brownian motion, we need to establish the following: (i) $w_{\lambda}(\theta)$ is continuous in $\lambda$, (ii) $w_{0}(\theta) = 0$ for all $\theta \in \Theta$, (iii) $w_{\lambda}(\theta)$ has independent increments, \emph{i.e.}, for 
$0 \le \lambda_1 \le \lambda_1' \le \lambda_2 \le \lambda_2'$, $w_{\lambda_1'}(\theta) - w_{\lambda_1}(\theta)$ and $w_{\lambda_2'}(\theta) - w_{\lambda_2}(\theta)$ are independent random variables, and (iv) for $\lambda_1>\lambda_2$, $w_{\lambda_1}(\theta) - w_{\lambda_2}(\theta) \sim {\cal N}(0,\lambda_1-\lambda_2)$.  The first three points follow from Theorem \ref{thm:telescoping_rec_gmrf_unit_disc}.  To show the last point, let $\theta_1=\theta_2$ in (\ref{eq:w_cov_l_2}) and use the computations done in (\ref{eq:cov_1})-(\ref{eq:cov_4}).
\end{IEEEproof}
Theorem~\ref{thm:w_l_noise} says that for each fixed $\theta$, $w_{\lambda}(\theta)$ in (\ref{eq:telescoping_unit_disc}) is Brownian motion.  This is extremely useful since we can use standard Ito calculus to interpret (\ref{eq:telescoping_unit_disc}).

\subsubsection{\textbf{White noise}}
A useful interpretation of $w_{\lambda}(\theta)$ is in terms of white noise.  Define a random field $v_{\lambda}(\theta)$ such that
\begin{equation}
w_{\lambda}(\theta) = \int_0^{\lambda} v_{\gamma}(\theta) d \gamma \,.
\end{equation}
Using Theorem \ref{thm:telescoping_rec_gmrf_unit_disc}, we can easily establish that $v_{\lambda}(\theta)$ is a generalized process such that for an appropriate function $\Psi(\cdot)$,
\begin{equation}
\int_0^{1} \Psi(\gamma) E[v_{\gamma}(\theta_1) v_{\lambda}(\theta_2)] d \gamma = \Psi(\lambda) \frac{C_{\lambda}(\theta_1,\theta_2)}{B_{\lambda}(\theta_1)B_{\lambda}(\theta_2)} \,,
\end{equation}
which is equivalent to the expression
\begin{equation}
E[v_{\lambda_1}(\theta_1) v_{\lambda_2}(\theta_2)] = \delta(\lambda_1-\lambda_2) \frac{C_{\lambda_1}(\theta_1,\theta_2)}{B_{\lambda_1}(\theta_1)B_{\lambda_2}(\theta_2)} \,.
\end{equation}
Using the white noise representation, an alternative form of the telescoping representation is given by
\begin{equation}
\frac{d x_{\lambda}(\theta)}{d\lambda}  = F_{\theta} [x_{\lambda}(\theta)] + B_{\lambda}(\theta) v_{\lambda}(\theta) \,.
\end{equation}

\subsubsection{\textbf{Boundary Conditions}}
From the form of the integral transform $F_{\theta}$ in (\ref{eq:f_theta_unit_disc}), it is clear that boundary conditions for the telescoping representation will be given in terms of the field defined at the boundary and its normal derivatives.  A general form for the boundary conditions can be given as
\begin{align}
\sum_{j=0}^{m-1} \int_{\Theta} c_{k,j}(\theta,\alpha) 
\frac{\partial^j}{\partial n^j} x_0(\alpha) d \alpha = \beta_k(\theta) \,, \nonumber \\
\theta \in \Theta \,, k = 1,\ldots, m \,, \label{eq:boundary_conditions}
\end{align}
where for each $k$, $\beta_k(\theta)$ is a Gaussian process in $\theta$ with mean zero and known covariance.

\subsubsection{\textbf{Integral Form}}
The representation in (\ref{eq:telescoping_unit_disc}) is a symbolic representation for the equation
\begin{align}
x_{\lambda_1}(\theta) &= x_{\lambda_2}(\theta) + 
\int_{\lambda_2}^{\lambda_1} F_{\theta} [x_{\mu}(\theta)] d\mu \nonumber \\
&\hspace{1.5cm}+ \int_{\lambda_2}^{\lambda_1} B_{\mu}(\theta) dw_{\mu}(\theta) \,,
\quad \lambda_1 > \lambda_2 \,. \label{eq:integral_form_unit_disc}
\end{align}
Since from Theorem \ref{thm:w_l_noise}, $w_{\lambda}(\theta)$ is Brownian motion for fixed $\theta$, the last integral in (\ref{eq:integral_form_unit_disc}) is an Ito integral.
Thus, to recursively synthesize the field, we start with boundary values, given by (\ref{eq:boundary_conditions}), and generate the field values recursively on the telescoping surfaces $\partial T^{\lambda}$ for $\lambda \in (0,1]$.

\subsubsection{\textbf{Comparison to \cite{TewfikLevyWillsky1991}}}  The telescoping recursive representation differs significantly from the recursive representation derived in \cite{TewfikLevyWillsky1991}.  Firstly, the representation in \cite{TewfikLevyWillsky1991} is only valid for isotropic GMRFs and does not hold for nonisotropic GMRFs.  The telescoping representation we derive holds for arbitrary GMRFs.  Secondly, the recursive representation in \cite{TewfikLevyWillsky1991} was derived on the Fourier series coefficients, whereas we derive a representation directly on the field values.

\subsection{Homogeneous and Isotropic GMRFs}
\label{sec:telescoping_example}

In this Section, we study homogeneous isotropic random fields over $\R^2$ whose covariance only depends on the Euclidean distance between two points.  In general, suppose $R_{\mu,\lambda}(\theta_1,\theta_2)$ is the covariance of a homogeneous isotropic random field over a unit disc such that the point $(\mu,\theta_1)$ in polar coordinates corresponds to the point $((1-\mu)\cos \theta,(1-\mu) \sin \theta)$ in Cartesian coordinates.  The Euclidean distance between two points $(\mu,\theta_1)$ and $(\lambda,\theta_2)$ is given by
\begin{align}
D_{\mu,\lambda}(\theta_1,\theta_2)
&= \left[ (1-\mu)^2 + (1-\lambda)^2 \nonumber \right.\\
&\left.- 2(1-\mu)(1-\lambda) \cos(\theta_1 - \theta_2) \right]^{1/2} \,.
\end{align}
If $R_{\mu,\lambda}(\theta_1,\theta_2)$ is the covariance of a homogeneous and isotropic GMRF, we have
\begin{equation}
R_{\mu,\lambda}(\theta_1,\theta_2) =
\Upsilon\left(D_{\mu,\lambda}(\theta_1,\theta_2)\right) \,, \label{eq:iso_cov}
\end{equation}
where $\Upsilon(\cdot) : \R \rightarrow \R$ is assumed to be differentiable at all points in $\R$.  The next Lemma computes $C_{\lambda}(\theta_1,\theta_2)$ for isotropic and homogeneous GMRFs.

\begin{lemma}
\label{lemma:h_i}
For an isotropic and homogeneous GMRF with covariance given by (\ref{eq:iso_cov}), $C_{\lambda}(\theta_1,\theta_2)$, defined in (\ref{eq:c_t_unit_disc}), is given by
\begin{equation}
C_{\lambda}(\theta_1,\theta_2) =
\left\{
\begin{array}{cc}
0 & \theta_1 \ne \theta_2 \\
-2 \Upsilon'(0) & \theta_1 = \theta_2 
\end{array} \,.\right. \label{eq:c_homo_iso}
\end{equation}
\end{lemma}
\begin{IEEEproof}
For $\theta_1 \ne \theta_2$, we have
\begin{align}
\frac{\partial}{\partial u} R_{\mu,\lambda}(\theta_1,\theta_2) & \nonumber \\
&\hspace{-2.5cm}= -\Upsilon'\left(D_{\mu,\lambda}(\theta_1,\theta_2)\right) 
\frac{(1-\mu) - (1-\lambda)\cos(\theta_1-\theta_2)}
{D_{\mu,\lambda}(\theta_1,\theta_2)} \,,
\end{align}
where $\Upsilon'(\cdot)$ is the derivative of the function $\Upsilon(\cdot)$.  Using (\ref{eq:c_t_unit_disc}), $C_{\lambda}(\theta_1,\theta_2) = 0$ when $\theta_1 \ne \theta_2$.

For $\theta_1 = \theta_2$, $D_{\mu,\lambda}(\theta_1,\theta_2) = |\mu - \lambda|$, so we have
\begin{equation}
\frac{\partial}{\partial u} R_{\mu,\lambda}(\theta_1,\theta_2)
= -\Upsilon'(|\lambda - \mu|) 
\frac{\lambda - \mu}{|\lambda - \mu|} \,,
\end{equation}
Using (\ref{eq:c_t_unit_disc}), $C_{\lambda}(\theta_1,\theta_2) = -2\Upsilon'(0)$ when $\theta_1 = \theta_2$.  
\end{IEEEproof}

Using Lemma~\ref{lemma:h_i}, we have the following theorem regarding the driving noise $w_{\lambda}(\theta)$ of the telescoping representation of an isotropic and homogeneous GMRF.

\begin{theorem}[Homogeneous isotropic GMRFs]
\label{eq:tele_homo_iso}
\label{thm:isotropic_gmrf}
For homogeneous isotropic GMRFs, with covariance given by (\ref{eq:iso_cov}), such that $\Upsilon(\cdot)$ is differentiable at all points in $\R$ and $\Upsilon'(0) < 0$, the telescoping representation is
\begin{equation}
dw_{\lambda}(\theta) = F_{\theta} x_{\lambda}(\theta) + \sqrt{-\Upsilon'(0)} d w_{\lambda}(\theta) \,. \label{eq:w_h_i}
\end{equation}
For each fixed $\theta$, $w_{\lambda}(\theta)$ is Brownian motion in $\lambda$ and 
\begin{align}
E[w_{\lambda_1}(\theta_1)w_{\lambda_2}(\theta_2)] &= 0\,, 
\quad \lambda_1 \ne \lambda_2, \theta_1 \ne \theta_2 \label{eq:w_p_1}\\
E[w_{\lambda}(\theta_1)w_{\lambda}(\theta_2)] &= 0\,,
\quad \theta_1 \ne \theta_2 \,. \label{eq:w_p_2}
\end{align}
\end{theorem}
\begin{IEEEproof}
Since we assume $\Upsilon'(0) < 0$, thus $B_{\lambda}(\theta) = C_{\lambda}(\theta,\theta) =  \sqrt{-\Upsilon'(0)}$, which gives us (\ref{eq:w_h_i}).  To show (\ref{eq:w_p_1}) and (\ref{eq:w_p_2}), we simply substitute the value of $C_{\lambda}(\theta_1,\theta_2)$, given by (\ref{eq:c_homo_iso}), in (\ref{eq:w_cov_l_1}) and use the independent increments property of $w_{\lambda}(\theta)$ given in Theorem~\ref{thm:telescoping_rec_gmrf_unit_disc}.
\end{IEEEproof}

\noindent
\textbf{Example:} 
We now consider an example of a homogeneous and isotropic GMRF where $\Upsilon'(0) = 0$ and thus the field is uniquely determined by the boundary conditions.
Let $\Upsilon(t)$, $t \in [0,\infty)$, be such that
\begin{equation}
\Upsilon(t) = \int_0^{\infty} \frac{b}{(1+b^2)^2} J_0(b t) db \,,
\end{equation}
where $J_n(\cdot)$ is the Bessel function of the first kind of order $n$ \cite{Watson1995}.  The derivative of $\Upsilon(t)$ is given by
\begin{equation}
\Upsilon'(t) = -\int_0^{\infty} \frac{b^2}{(1+b^2)^2} J_1(b t) db \,,
\end{equation}
where we use the fact that $J_0'(\cdot) = - J_1(\cdot)$ \cite{Watson1995}.  Since $J_1(0) = 0$, $\Upsilon'(0) = 0$ and thus $B_{\lambda}(\theta) = 0$ in the telescoping representation.  This means there is no driving noise in the telescoping representation.  The rest of the parameters of the telescoping representation can be computed using the fact \cite{Wong1968}
\begin{equation}
(\Delta - 1)^2 R(t,s) = \delta(t-s) \label{eq:lr=d_2}\,,
\end{equation}
where $R(t,s)$ corresponds to the covariance associated with $\Upsilon(\cdot)$ written in Cartesian coordinates and $\Delta$ is the Laplacian operator.  Since the operator associated with $R(t,s)$ in (\ref{eq:lr=d_2}) has order four, it is clear that the GMRF has order two. The field with covariance satisfying (\ref{eq:lr=d_2}) is also commonly referred to as the Whittle field \cite{Whittle1954}.  The telescoping recursive representation will be of the form
\begin{align}
dx_{\lambda}(\theta) &= \int_{-\pi}^{\pi}
\lim_{\mu \rightarrow \lambda^+}
\left[
\frac{\partial}{\partial u} b_0((\mu,\theta),(\lambda,\alpha)) x_{\lambda}(\alpha) \right. \\
&\hspace{-1cm}\left.+ \frac{\partial}{\partial u} b_1((\mu,\theta),(\lambda,\alpha)) \frac{\partial}{\partial n}x_{\lambda}(\alpha) \right] d\alpha \,d \lambda \,,\quad 0 \le \lambda \le 1 \,, \nonumber 
\end{align}
with appropriate boundary conditions defined on the unit circle.

\section{Telescoping Representation: GMRFs on arbitrary domains}
\label{sec:telescoping_cont_arbit}
In the last Section, we presented telescoping recursive representations for random fields defined on a unit disc.  In this Section, we generalize the telescoping representations to arbitrary domains.  Section \ref{sec:telescoping_surface_homotopy} shows how to define telescoping surfaces using the concept of homotopy.  Section \ref{sec:paramet_domains} shows how to parametrize arbitrary domains using the homotopy.  Section \ref{sec:telescoping_repr} presents the telescoping representation for GMRFs defined on arbitrary domains.

\subsection{Telescoping Surfaces Using Homotopy}
\label{sec:telescoping_surface_homotopy}

Informally, a homotopy is defined as a continuous deformation from one space to another.  Formally, given two continuous functions $f$ and $g$ such that $f,g: {\cal X} \rightarrow {\cal Y}$, a homotopy is a continuous function $h:{\cal X} \times [0,1] \rightarrow {\cal Y}$ such that if $x \in {\cal X}$, $h(x,0) = f(x)$ and $h(x,1) = g(x)$ \cite{Munkres2000}.  An example of the use of homotopy in neural networks is shown in \cite{CoetzeeStonick}.  

In deriving our telescoping representation for GMRFs on a unit disc in Section \ref{sec:telescoping_cont}, we saw that the recursions started at the boundary, which was the unit circle, and telescoped inwards on concentric circles and ultimately converged to the center of the unit disc.  To parametrize these recursions, we can define a homotopy from the unit circle to the center of the unit disc.  In general, for a domain $T \subset \R^d$ with smooth boundary $\partial T$, the telescoping surfaces can be defined using a homotopy, $h:\partial T \times [0,1] \rightarrow c$, from the boundary $\partial T$ to a point $c \in T$ such that
\begin{enumerate}[P1.]
\item $\{h(t,0) :  t \in \partial T\} = \partial T$ and $\{h(t,1) :  t \in \partial T\} = c$.
\item For $0 < \lambda \le 1$, $\{h(t,\lambda) : t \in \partial T\} \subset T$ is the boundary of the region 
$\{h(t,\mu) : (t,\mu) \in \partial T \times (\lambda,1)\}$.
\item For $\lambda_1 < \lambda_2$,
$\{h(t,\lambda_1), t \in \partial T\} \subset \{h(t,\mu), t \in \partial T, 0 \le \mu \le \lambda_2 \}$.
\item $ \bigcup_{\lambda} \{h(t,\lambda), t \in \partial T \} = T \,.$
\end{enumerate}
Property 1 says that, for $\lambda = 0$, we get the boundary $\partial T$ and for $\lambda = 1$, we get the point $c \in T$, which we choose arbitrarily.  Property 2 says that for each $\lambda$, we want the telescoping surfaces to be in $T$ and it should be a boundary of another region.  Property 3 restricts the surfaces to be contained within each other, and Property 4 says that the homotopy must sweep the whole index set $T$.

\begin{figure}
\begin{center}
\subfigure[]{
\includegraphics[scale=0.2]{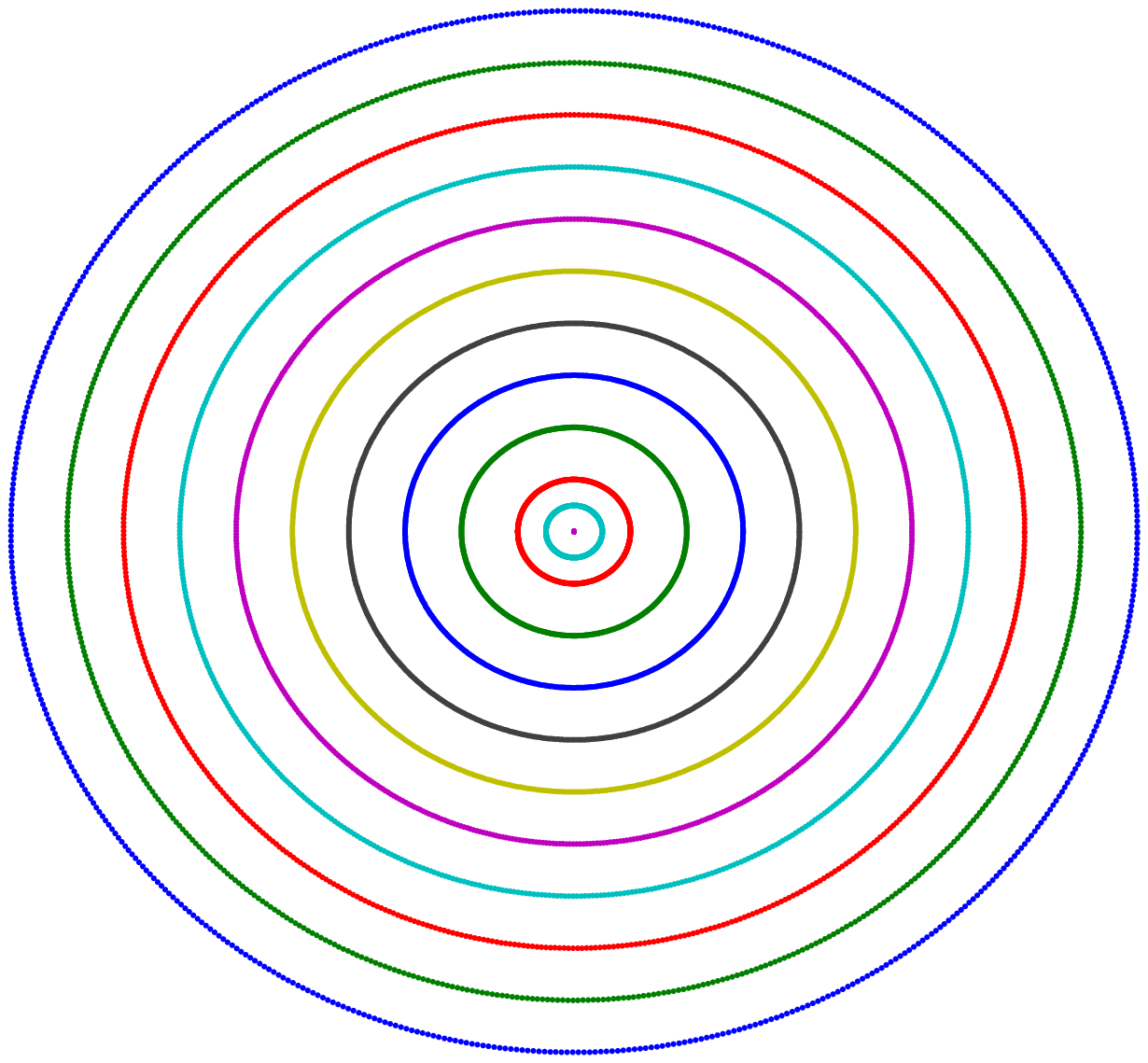}
}
\subfigure[]{
\includegraphics[scale=0.2]{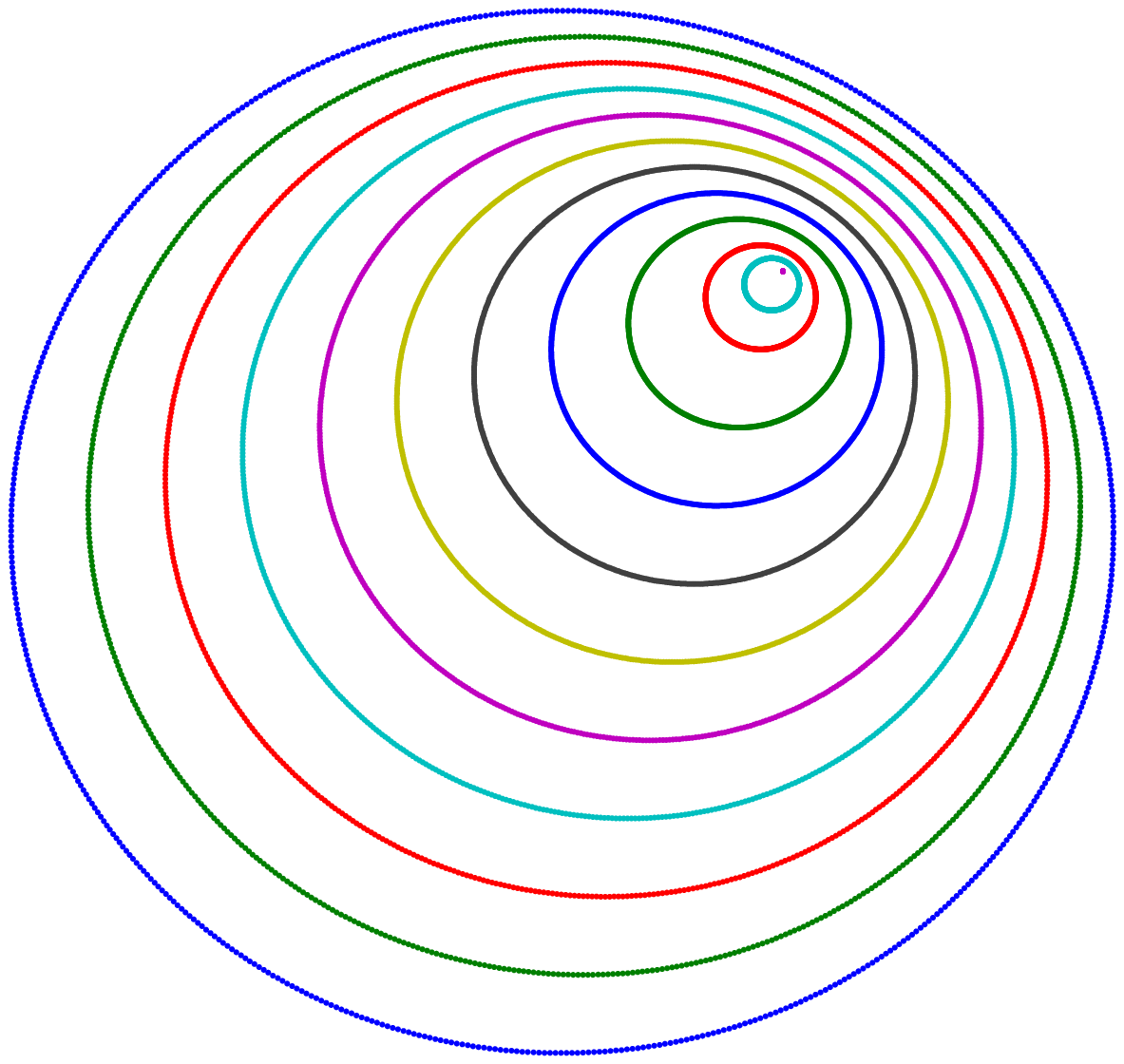}
}
\subfigure[]{
\includegraphics[scale=0.2]{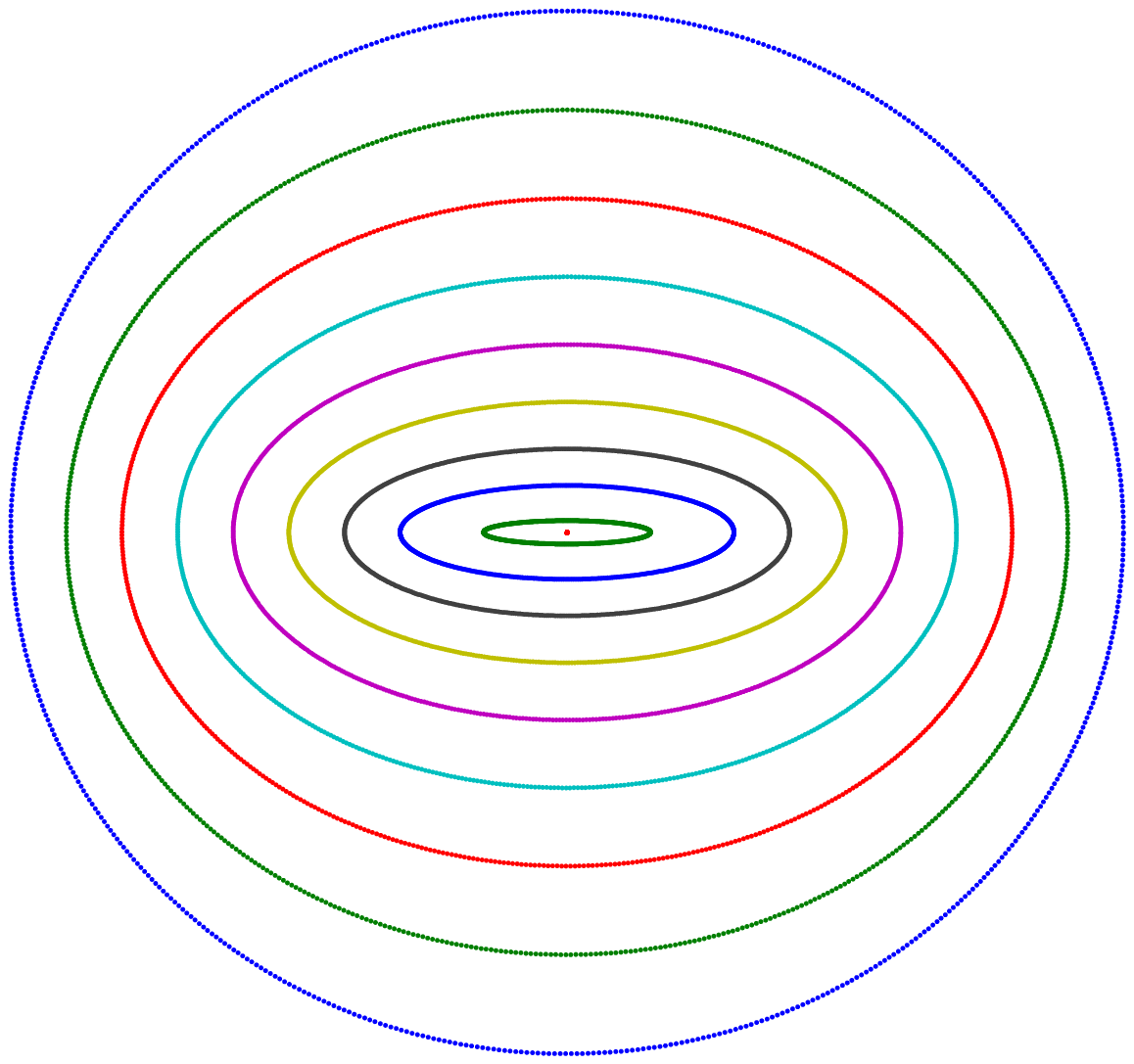}
}
\subfigure[]{
\includegraphics[scale=0.2]{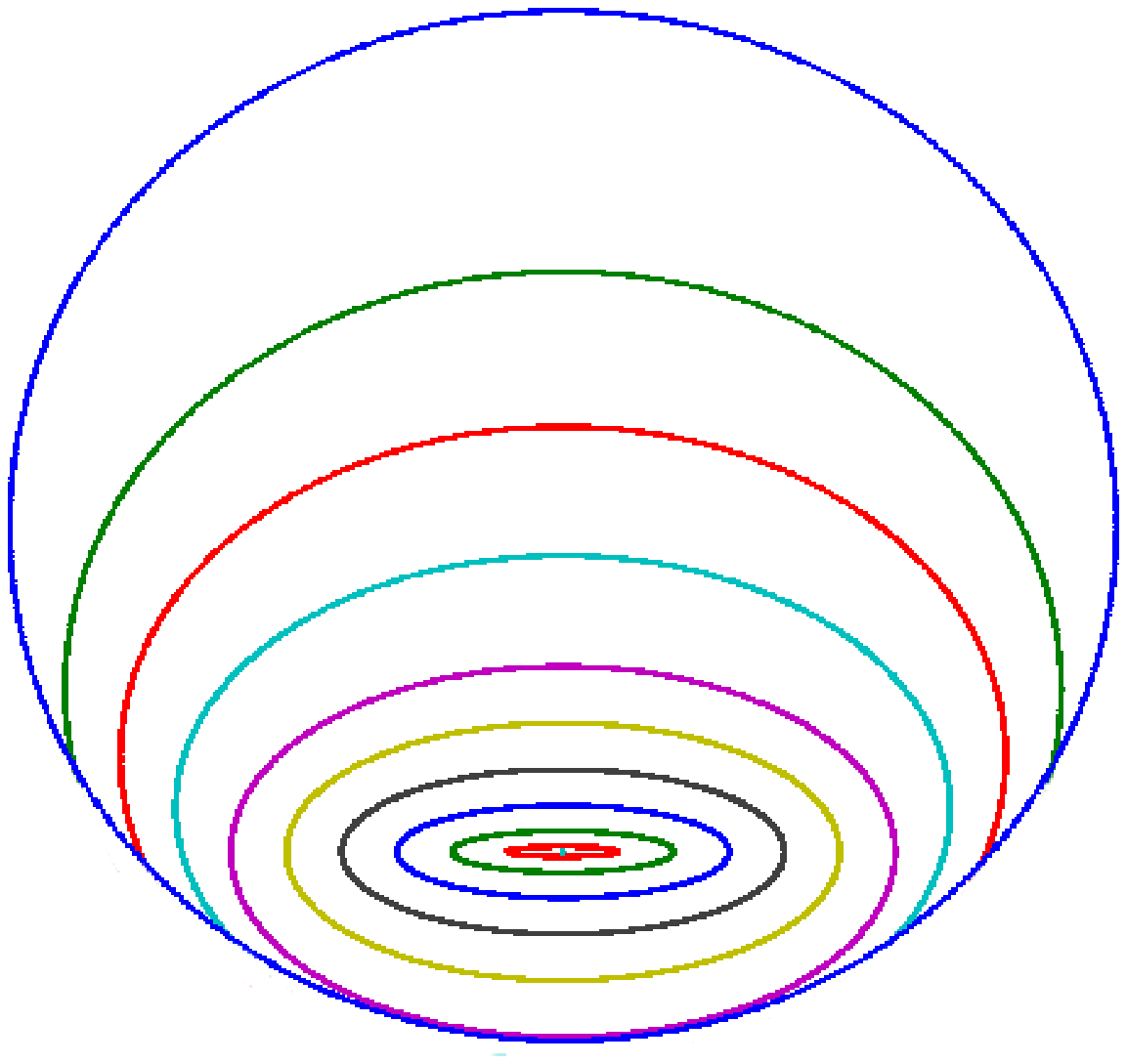}
}
\caption{Telescoping Surfaces defined using different homotopies.}
\label{fig:telescoping_surfaces}
\end{center}
\end{figure}

Using the homotopy, for each $\lambda$, we can define a telescoping surface $\partial T^{\lambda}$ such that
\begin{equation}
\partial T^{\lambda} = \{h(\theta,\lambda) : \theta \in \partial T\} \,, \label{eq:homotopy_to_boundary}
\end{equation}
where $\partial T$ is the boundary of the field.  As an example, we consider defining different telescoping surfaces for the field defined on a unit disc.  The boundary of the unit disc can be parametrized by the set of points
\begin{equation}
\partial T = \{(\cos \theta,\sin \theta), \theta \in [-\pi,\pi]\} \,.
\end{equation}
We consider four different kinds of telescoping surfaces:
\begin{enumerate}[a)]
\item The telescoping surfaces in Fig~\ref{fig:telescoping_surfaces}(a) are generated using the homotopy
\begin{equation}
h((\cos\theta,\sin\theta),\lambda) = ((1-\lambda)\cos \theta,(1-\lambda)\sin \theta) \,. \label{eq:conc_circles}
\end{equation}

\item The telescoping surfaces in Fig~\ref{fig:telescoping_surfaces}(b) can be generated by the homotopy

$h((\cos\theta,\sin\theta),\lambda) $
\begin{align}
&= ((1-\lambda)(\cos \theta - c_1) + c_1,
(1-\lambda) (\sin \theta - c_2) + c_2) \,,\\
&= ((1-\lambda)\cos \theta + c_1 - (1-\lambda)c_1, \label{eq:nonsym_circles} \\
&\hspace{2cm} (1-\lambda)\cos \theta + c_2 - (1-\lambda)c_2) \,, 
\nonumber
\end{align}
where $(c_1,c_2)$ is inside the unit disc, \emph{i.e.}, $c_1^2 + c_2^2 < 1$.  For the homotopy in (\ref{eq:conc_circles}), each telescoping surface is centered about the origin, whereas the telescoping surfaces in (\ref{eq:nonsym_circles}) are centered about the point
$(c_1-(1-\lambda)c_1,c_2-(1-\lambda)c_2)$.
\item In Fig~\ref{fig:telescoping_surfaces}(a)-(b), the telescoping surfaces are circles, however, we can also have other shapes for the telescoping surface.  Fig~\ref{fig:telescoping_surfaces}(c) shows an example in which the telescoping surface is an ellipse, which we generate using the homotopy
\begin{equation}
h((\cos\theta,\sin\theta),\lambda) =
(a_{\lambda} \cos \theta, b_{\lambda} \sin \theta) \\
\end{equation}
where $a_{\lambda}$ and $b_{\lambda}$ are continuous functions chosen in such a way that P1-P4 are satisfied for $h$.  In Fig~\ref{fig:telescoping_surfaces}(c), we choose $a_{\lambda} = \lambda$ and $b_{\lambda} = \lambda^2$.

\item Another example of a set of telescoping surfaces is shown in Fig~\ref{fig:telescoping_surfaces}(d).  From here, we notice that two telescoping surfaces may have common points.
\end{enumerate}

Apart from the telescoping surfaces for a unit disc shown in Fig~\ref{fig:telescoping_surfaces}(a)-(d), we can define many more telescoping surfaces.  The basic idea in obtaining these surfaces, which is compactly captured by defining a homotopy, is to continuously deform the boundary of the index set until we converge to a point within the index set.  In the next Section, we provide a characterization of continuous index sets in $\R^d$ for which we can easily find telescoping surfaces by simply scaling and translating the points on the boundary.

\subsection{Generating Similar Telescoping Surfaces}
\label{sec:paramet_domains}

From Section \ref{sec:telescoping_surface_homotopy}, it is clear that, for a given domain, many different telescoping surfaces can be obtained by defining different homotopies.  In this Section, we identify domains on which we can easily generate a set of telescoping surfaces, which we call \emph{similar telescoping surfaces}.

\begin{definition}[Similar Telescoping Surfaces]
\label{def:similar_telescoping_surfaces}
Two telescoping surfaces are \emph{similar} if there exists an affine map between them, \emph{i.e.}, we can map one to another by scaling and translating of the coordinates.  A set of telescoping surfaces are similar if each pair of telescoping surfaces in the set are similar.
\end{definition}

As an example, the set of telescoping surfaces in Fig~\ref{fig:telescoping_surfaces}(a)-(b) are similar since all the telescoping surfaces are circles.  On the other hand, the telescoping surfaces in Fig~\ref{fig:telescoping_surfaces}(c)-(d) are not similar since each telescoping surfaces has a different shape.  The following theorem shows that, for certain index sets, we can always find a set of similar telescoping surfaces.

\begin{theorem}
For a domain $T \in \R^d$ with boundary $\partial T$ if there exists a point $c \in T$ such that, for all $t \in T \cup \partial T$ and $\lambda \in [0,1]$, 
$(1-\lambda) t + \lambda c \in T$, we can generate similar telescoping surfaces using the homotopy
\begin{equation}
h(\theta,\lambda) = (1-\lambda) \theta + \lambda c \,,\quad \theta \in \partial T \,. \label{eq:homotopy_similar}
\end{equation}
\end{theorem}
\begin{IEEEproof}
Given the homotopy in (\ref{eq:homotopy_similar}), the telescoping surfaces are given by $\partial T^{\lambda} = \{h(\theta,\lambda) : \theta \in \partial T\}$.  Using (\ref{eq:homotopy_similar}), it is clear that $\partial T^0 = \partial T$ and $\partial T^1 = c$.  Given the assumption, we have that $\partial T^{\lambda} \subset T$ for $0 < \lambda \le 1$.  Since the distance of each point on $\partial T^{\lambda}$ to the point $c$ is $(1-\lambda)||\theta - c||$, it is clear that, for $\lambda_1 < \lambda_2$, 
$\partial T^{\lambda_1} \subset \{\partial T^{\mu}: 0 \le \mu \le  \lambda_2\}$.  This shows that the homotopy in (\ref{eq:homotopy_similar}) defines a valid telescoping surface.  The set of telescoping surfaces is similar since we are only scaling and translating the boundary $\partial T$.
\end{IEEEproof}

\begin{figure}
\begin{center}
\subfigure[]{
\includegraphics[scale=0.2]{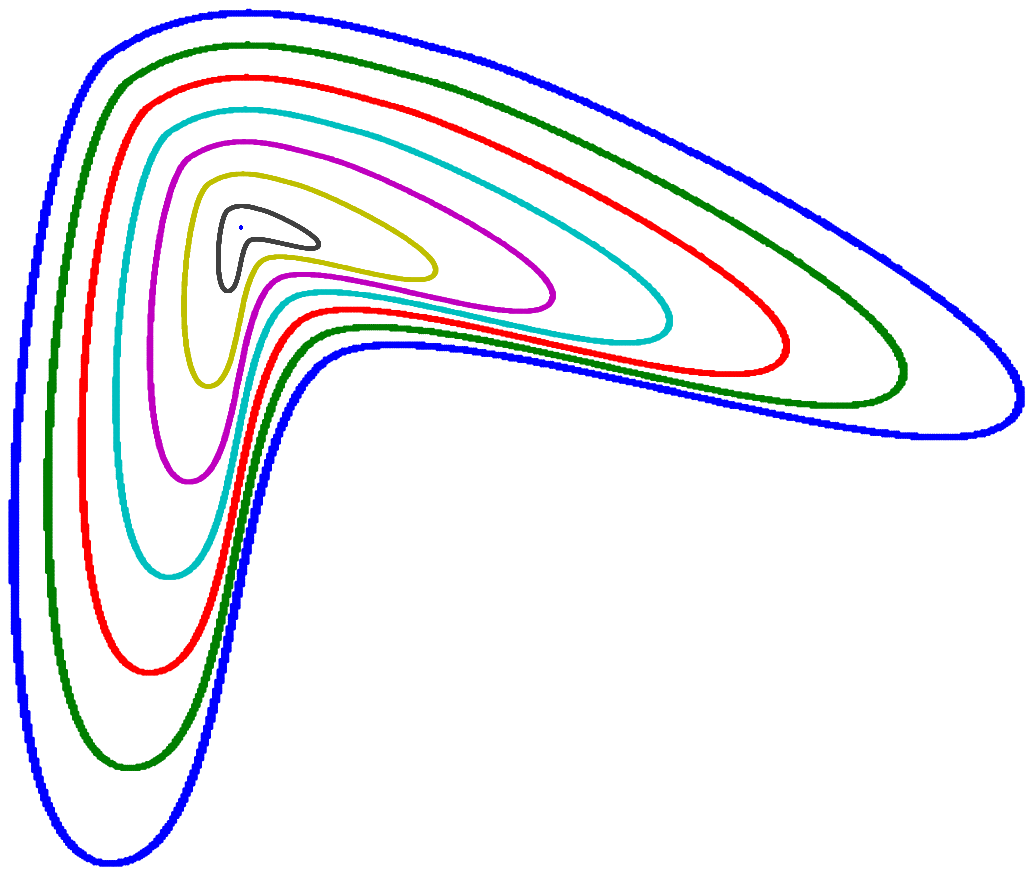}
}
\subfigure[]{
\includegraphics[scale=0.2]{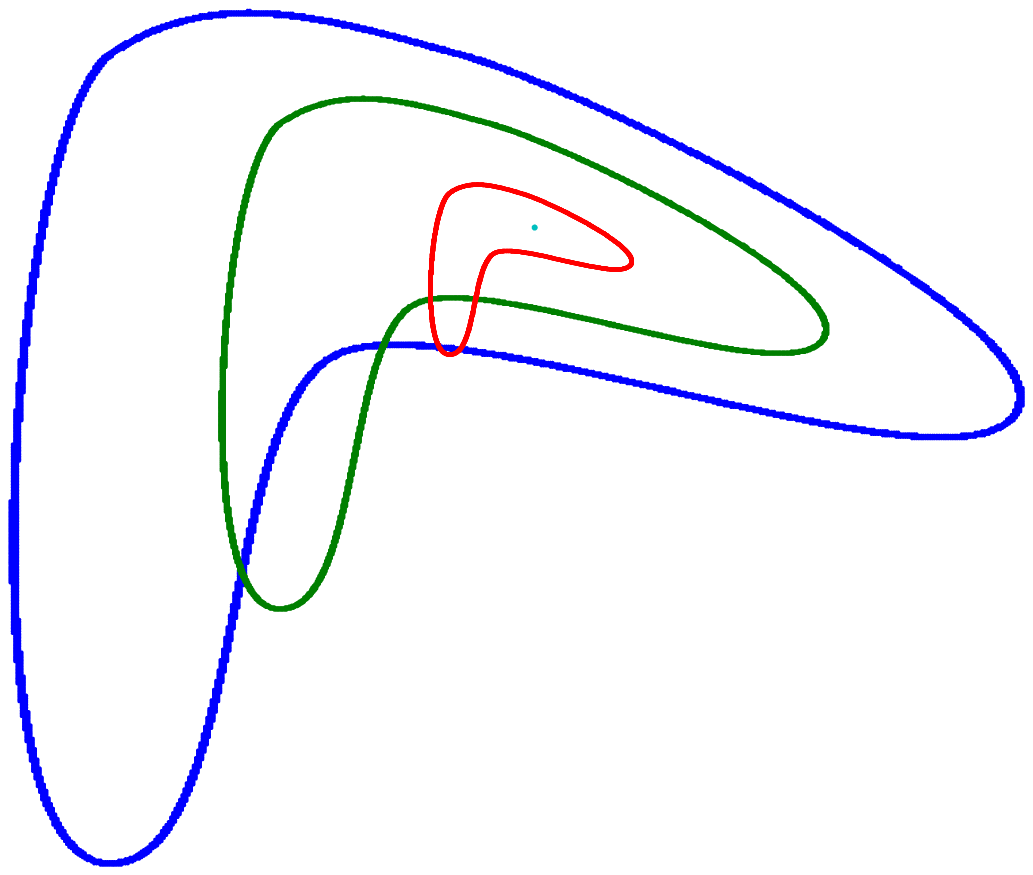}
}
\subfigure[]{
\includegraphics[scale=0.2]{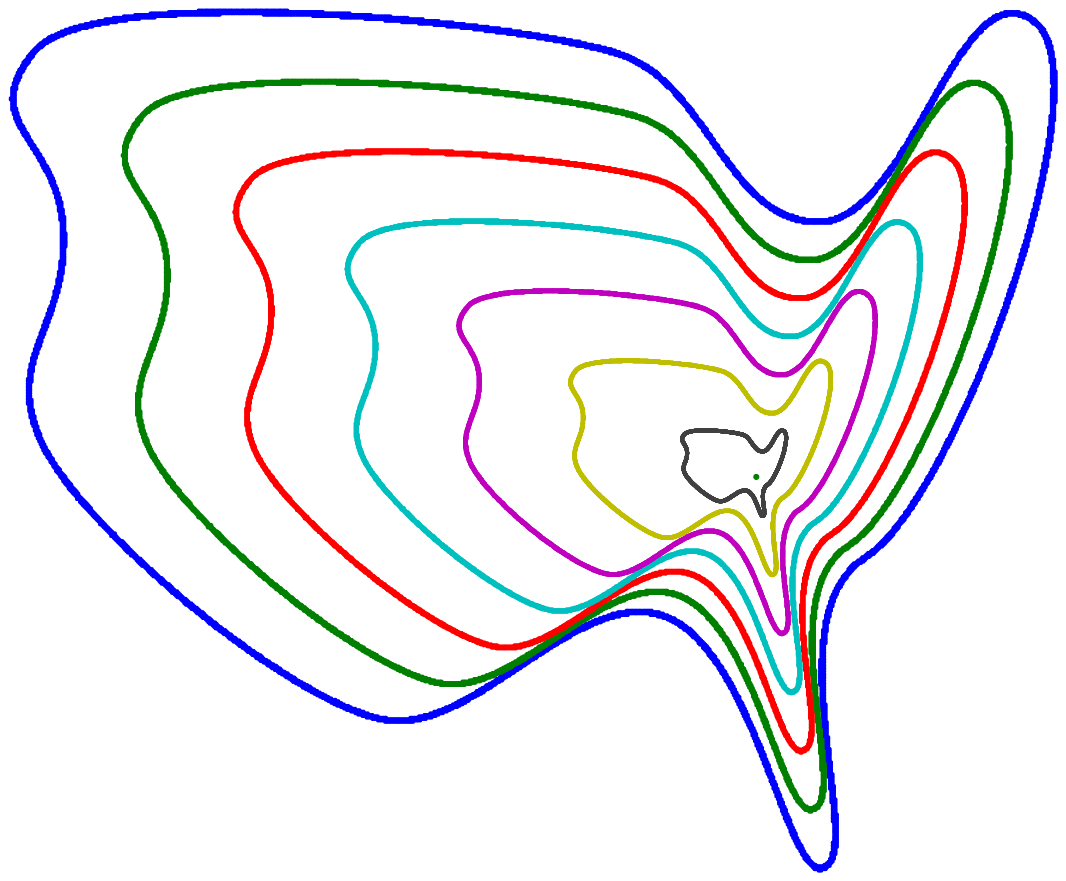}
}
\subfigure[]{
\includegraphics[scale=0.15]{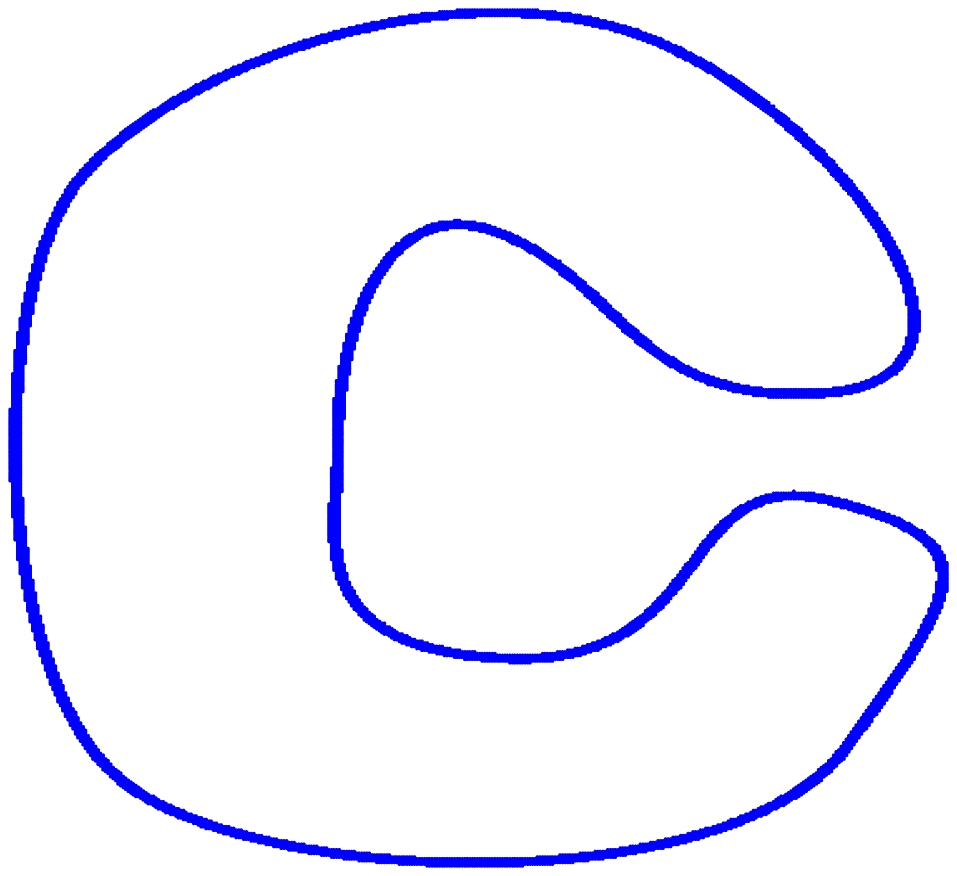}
}
\caption{Telescoping Surfaces defined using different homotopies.}
\label{fig:similar_telescoping_surfaces}
\end{center}
\end{figure}

Examples of similar telescoping surfaces generated using the homotopy in (\ref{eq:homotopy_similar}) are shown in Fig~\ref{fig:similar_telescoping_surfaces}(a) and Fig~\ref{fig:similar_telescoping_surfaces}(c).  Choosing an appropriate $c$ is important to generate similar telescoping surfaces.  For example, Fig~\ref{fig:similar_telescoping_surfaces}(b) shows an example where telescoping surfaces are generated using (\ref{eq:homotopy_similar}).  It is clear that these surfaces do not satisfy the desired properties of telescoping surfaces.
Fig~\ref{fig:similar_telescoping_surfaces}(d) shows an example of an index set for which similar telescoping surfaces do not exist since there exists no point $c$ for which $(1-\lambda)t + \lambda c \in T$ for all $\lambda \in [0,1]$ and $t \in T \cup \partial T$.

\subsection{Telescoping Representations}
\label{sec:telescoping_repr}

We now generalize the telescoping representation to GMRFs defined on arbitrary domains.  Let $x(t)$ be a zero mean GMRF, where $t \in T \subset \R^d$ such that the smooth boundary of $T$ is $\partial T$.  Define a set of telescoping surfaces $\partial T^{\lambda}$ constructed by defining a homotopy $h(\theta,\lambda)$, where $\theta \in \partial T$ and $\lambda \in [0,1]$.  We parametrize the GMRF $x(t)$ as $x_{\lambda}(\theta)$ such that
\begin{equation}
x_{\lambda}(\theta) = x(h(\theta,\lambda)) \,.
\end{equation}
Denote $\Theta = \partial T$ and define $C_{\lambda}(\theta_1,\theta_2)$, $B_{\lambda}(\theta)$, and $F_{\theta}$ by (\ref{eq:c_t_unit_disc}), (\ref{eq:b_l_unit_disc}), and (\ref{eq:f_theta_unit_disc}), respectively.  
Although the initial definition for these values was for $\Theta = [-\pi,\pi]$ and $x_{\lambda}(\theta)$ parametrized in polar coordinates, assume the definitions in (\ref{eq:c_t_unit_disc}), (\ref{eq:b_l_unit_disc}), and (\ref{eq:f_theta_unit_disc}) are in terms of the parameters defined in this Section.  The normal derivatives in the definition of $F_{\theta}$ for a point $x_{\lambda}(\theta)$ will be computed in the direction normal to the telescoping surface $\partial T^{\lambda}$ at the point $h(\theta,\lambda)$.  
The telescoping representation is given by
\begin{equation}
dx_{\lambda}(\theta) = F_{\theta}[x_{\lambda}(\theta)] d\lambda + B_{\lambda}(\theta) dw_{\lambda}(\theta) \label{eq:telescoping_arbit}
\end{equation}
where the $w_{\lambda}(\theta) = w(h(\theta,\lambda))$ is the driving noise with the same properties as outlined in Theorem~\ref{thm:telescoping_rec_gmrf_unit_disc}.  It is clear from (\ref{eq:telescoping_arbit}), that the recursions for the GMRF initiate at the boundary and recurse inwards along the telescoping surfaces defined using the homotopy $h(\theta,\lambda)$.  Thus, the recursions are effectively captured by the parameter $\lambda$.

\section{Recursive Estimation of GMRFs}
\label{sec:rec_estimation_gmrf}

Using the telescoping representation, we now derive recursive equations for estimating GMRFs.  Let $x(t)$ be the zero mean GMRF defined on an index set $T \subset \R^d$ with smooth boundary $\partial T$.  Assume the parametrization $x_{\lambda}(\theta) = x(h(\theta,\lambda))$, where $h(\theta,\lambda)$ is an appropriate homotopy and $\theta \in \Theta = \partial T$.  The corresponding telescoping representation is given in (\ref{eq:telescoping_arbit}).

Consider the observations, written in parametric form, as
\begin{equation}
d y_{\lambda}(\theta) = G_{\lambda}(\theta) x_{\lambda}(\theta) d\lambda + D_{\lambda}(\theta) dn_{\lambda}(\theta) \,,\; 0\le\lambda\le 1 \,, \label{eq:observations}
\end{equation}
where $G_{\lambda}(\theta)$ and $D_{\lambda}(\theta)$ are known functions with $D_{\lambda}(\theta) \ne 0$, $y_0(\theta) = 0$ for all $\theta \in \Theta$, $n_{\lambda}(\theta)$ is standard Brownian motion for each fixed $\theta$ such that
\begin{equation}
E[n_{\lambda_1}(\theta_1) n_{\lambda_2}(\theta_2)] = 0\,,\quad
\lambda_1 \ne \lambda_2, \theta_1 \ne \theta_2 \,,
\end{equation}
and $n_{\lambda}(\theta)$ is independent GMRF $x_{\lambda}(\theta)$.

We consider the filtering and smoothing problem for GMRFs.  For random fields, because of the multidimensional index set, it is not clear how to define the filtered estimate.  For Markov processes, the filtered estimate sweeps the data in a causal manner, because the process itself admits a causal representation.  To define the filtered estimate for GMRFs, we sweep the observations over the telescoping surfaces defined in the telescoping recursive representation in (\ref{eq:telescoping_arbit}).
Define the filtered estimate $\widehat{x}_{\lambda|\lambda}(\theta)$, error $\widetilde{x}_{\lambda|\lambda}(\theta)$, and error covariance $S_{\lambda}(\alpha,\beta)$ such that
\begin{align}
\widehat{x}_{\lambda|\lambda}(\theta) &\triangleq 
E\left[ x_{\lambda}(\theta) | 
\sigma\{y_{\mu}(\theta) \,, 0\le \mu \le \lambda,\theta \in \Theta\} \right] \label{eq:filtered_est}\\
\widetilde{x}_{\lambda|\lambda}(\theta) &\triangleq 
{x}_{\lambda|\lambda}(\theta) -\widehat{x}_{\lambda|\lambda}(\theta) \label{eq:error}\\
S_{\lambda}(\alpha,\beta) &\triangleq 
E[\widetilde{x}_{\lambda | \lambda}(\alpha) 
\widetilde{x}_{\lambda | \lambda}(\beta)] \,.
\end{align}
The set $\{y_{\mu}(\theta) \,, 0\le \mu \le \lambda,\theta \in \Theta\}$ consists of the region between the boundary of the field, $\partial T$, and the surface $\partial T^{\lambda}$.  
A stochastic differential equation for the filtered estimate $\widehat{x}_{\lambda|\lambda}(\theta)$ is given in the following theorem.

\begin{theorem}[\textbf{Recursive Filtering of GMRFs}]
\label{thm:recursive_filter}
For the GMRF $x_{\lambda}(\theta)$ with observations $y_{\lambda}(\theta)$, a stochastic differential equation for the filtered estimate $\widehat{x}_{\lambda|\lambda}(\theta)$, defined in (\ref{eq:filtered_est}), is given as follows:
\begin{equation}
d \widehat{x}_{\lambda|\lambda}(\theta) = F_{\theta}[\widehat{x}_{\lambda|\lambda}(\theta)] d\lambda + 
K_{\theta} [de_{\lambda}(\theta)] \,, \label{eq:filter_field}
\end{equation}
where $e_{\lambda}(\theta)$ is the innovation field such that
\begin{equation}
D_{\lambda}(\theta) de_{\lambda}(\theta) = dy_{\lambda}(\theta) - 
G_{\lambda}(\theta) \widehat{x}_{\lambda | \lambda}(\theta) d\lambda\,,
\end{equation}
$F_{\theta}$ is the integral transform defined in (\ref{eq:f_theta_unit_disc}) and $K_{\theta}$ is an integral transform such that
\begin{equation}
K_{\theta} [de_{\lambda}(\theta)] = \int_{\Theta}
\frac{G_{\lambda}(\alpha)}{D_{\lambda}(\alpha)} S_{\lambda}(\alpha,\theta)
de_{\lambda}(\alpha) d \alpha \,, \label{eq:k_int}
\end{equation}
where $S_{\lambda}(\alpha,\theta)$ satisfies the equation
\begin{align}
\frac{\partial}{\partial \lambda}S_{\lambda}(\alpha,\theta) &= 
F_{\alpha} [S_{\lambda}(\alpha,\theta)] + 
F_{\theta} [S_{\lambda}(\alpha,\theta)] + 
C_{\lambda}(\theta,\alpha) \nonumber \\
&-\int_{\Theta} \frac{G_{\lambda}^2(\beta)}{D_{\lambda}^2(\beta)} 
S_{\lambda}(\alpha,\beta) S_{\lambda}(\theta,\beta) d\beta
\,. \label{eq:ricatti_s}
\end{align}
\end{theorem}
\begin{IEEEproof}
See Appendix \ref{appendix_proof_filter}.
\end{IEEEproof}
We show in Lemma~\ref{lemma:innovation} (Appendix B) that $e_{\lambda}(\theta)$ is Brownian motion.  Thus, (\ref{eq:filter_field}) can be interpreted using Ito calculus.  Since we do not observe the field on the boundary, we assume that the boundary conditions in (\ref{eq:filter_field}) are zero such that:
\begin{equation}
\frac{\partial^j}{\partial n^j} \widehat{x}_{\lambda | \lambda}(\theta) = 0 \,,\; j = 1\,\ldots, m-1\,,\; \theta \in \Theta \,.
\end{equation}
The boundary equations for the partial differential equation associated with the filtered error covariance is computed using the covariance of the field at the boundary such that
\begin{equation}
\frac{\partial^j}{\partial n^j} S_{0}(\alpha,\theta)
= \frac{\partial^j}{\partial n^j}
R_{0,0}(\alpha,\theta) \,,\; \alpha,\theta \in \Theta \,.
\end{equation}

The filtering equation in (\ref{eq:filter_field}) is similar to the Kalman-Bucy filtering equations derived for Gauss-Markov processes. 
The differences arise because of the telescoping surfaces.  Using (\ref{eq:filter_field}), let $\Theta$ be a single point instead of $[-\pi,\pi]$.  In this case, the integrals in (\ref{eq:f_theta_unit_disc}) and (\ref{eq:k_int}) disappear and we easily recover the Kalman-Bucy filter for Gauss-Markov processes.

Using the filtered estimates, we now derive equations for smoothing GMRFs.  Define the smoothed estimate $\widehat{x}_{\lambda|T}(\theta)$, error $\widetilde{x}_{\lambda|T}(\theta)$, and error covariance $S_{\lambda | T}(\alpha,\beta)$ as follows:
\begin{align}
\widehat{x}_{\lambda|T}(\theta) &\triangleq E[x_{\lambda|T} | \sigma\{y(T)\}] \\
\widetilde{x}_{\lambda|T}(\theta) &\triangleq x_{\lambda}(\theta) - \widehat{x}_{\lambda|T}(\theta) \\
S_{\lambda | T}(\alpha,\beta) &= E[\widetilde{x}_{\lambda|T}(\alpha)
\widetilde{x}_{\lambda|T}(\beta)] \,.
\end{align}
A recursive smoother for GMRFs, similar to the Rauch-Tung-Striebel (RTS) smoother, is given as follows.
\begin{theorem}[\textbf{Recursive Smoothing for GMRFs}]
\label{thm:recursive_smoother}
For the GMRF $x_{\lambda}(\theta)$, assuming $S_{\lambda}(\theta,\theta) > 0$, the smoothed estimate is the solution to the following stochastic differential equation:
\begin{align}
d \widehat{x}_{\lambda|T}(\theta) = F_{\theta} [\widehat{x}_{\lambda|T}(\theta)] d \lambda &+ 
\frac{C_{\lambda}(\theta,\theta)}{S_{\lambda}(\theta,\theta)}
[\widehat{x}_{\lambda | T}(\theta) - \widehat{x}_{\lambda | \lambda}(\theta)]\,, \nonumber\\
& \hspace{2cm} 1 \ge \lambda \ge 0 \,,
\label{eq:smoother}
\end{align}
where $\widehat{x}_{\lambda | \lambda}(\theta)$ is calculated using Theorem \ref{thm:recursive_filter} and the smoother error covariance is a solution to the partial differential equation,
\begin{align}
\frac{\partial S_{\lambda | T} (\alpha,\theta)}{\partial \lambda} &= 
\widetilde{F}_{\theta}S_{\lambda | T} (\alpha,\theta) +
\widetilde{F}_{\alpha} S_{\lambda | T} (\alpha,\theta)
+ C_{\lambda}(\alpha,\theta) \nonumber \\
&\hspace{-0.5cm}-\frac{C_{\lambda}(\alpha,\alpha)S_{\lambda}(\alpha,\theta)}{S_{\lambda}(\alpha,\alpha)}
-\frac{C_{\lambda}(\theta,\theta)S_{\lambda}(\alpha,\theta)}{S_{\lambda}(\theta,\theta)} \,,
\label{eq:smoother_cov}
\end{align}
where 
\begin{equation}
\widetilde{F}_{\beta} = F_{\beta} + C_{\lambda}(\beta,\beta)/S_{\lambda}(\beta,\beta) \,.
\end{equation}
\end{theorem}
\begin{IEEEproof}
See Appendix \ref{app:proof_smoothing}
\end{IEEEproof}

The equations in Theorem \ref{thm:recursive_smoother} are similar to the Rauch-Tung-Striebel smoothing equations for Gauss-Markov processes \cite{RauchTungStriebel1965}.  Other smoothing equations for Gauss-Markov processes can be extended to apply to GMRFs.

\section{Telescoping Representations of GMRFS: Discrete Indices}
\label{sec:tel_rep_gmrf}

We now describe the telescoping representation for GMRFs when the index set is discrete.  For simplicity, we restrict the presentation to GMRFs with order two.  Let $\{x(i,j) \in \R: (i,j) \in T_1 = [1,N]\times [1,M] \}$ be the GMRF.  Stack each row of the field and form an $NM \times 1$ vector $\bx$.  In the representation (\ref{eq:mmse}), stack the noise field $v(i,j)$ row wise  into an $NM \times 1$ vector $\bv$.
The boundary values are indexed in a clockwise manner starting at the upper leftmost node into a $2(N+1)+2(M+1) \times 1$ vector $\bx_b$ \footnote{The ordering does not matter as long as the ordering is known.}.  A matrix equivalent of (\ref{eq:mmse}) is given as
\begin{equation}
{\cal A} \bx = {\cal A}_b \bx_b + \bv \,, \label{eq:matrix}
\end{equation}
where ${\cal A}$ is an $NM \times NM$ block tridiagonal matrix with block size $N \times N$, ${\cal A}_b$ is an $NM \times 2(N+1)+2(M+1)$ sparse matrix corresponding to the interaction of the nodes in $\bx$ with the boundary nodes.  The matrices ${\cal A}$ and ${\cal A}_b$ can be evaluated from the nonrecursive equation given in (\ref{eq:mmse}).  Further, we have the following relationships,
\begin{equation}
E[\bx \bv^T] = I \text{ and } E[\bv \bv^T] = {\cal A} \label{eq:vv_A}\,.
\end{equation}
Equation (\ref{eq:matrix}) is an extension of the matrix representation given in \cite{MouraBalram1992} for the case when $\bx_b = 0$, i.e., boundary conditions are zero.  For more properties about the structure of the matrix ${\cal A}$, we refer to \cite{MouraBalram1992}.

Let $T_k = [k,N+1-k] \times [k,M+1-k]$ and let $\partial T_k$ be the boundary nodes of the index set $T_k$ ordered in a clockwise direction.  For example, $x(\partial T_0) = \bx_b$.  Define $z_k$ such that
\begin{equation}
z_k \triangleq x(\partial T_k) = x(T_k \backslash T_{k+1}) \,, \label{eq:def_z_k}
\end{equation}
where $k = 0,1,\ldots,{\left\lceil \min(M,N)/2\right\rceil}$.  Define $\tau$ such that
\begin{equation}
\tau = {\left\lceil \min(M,N)/2\right\rceil} \,.
\end{equation}
Each $z_k$ will be of variable size, and let $M^z_k$ be the size of $z_k$, i.e., $z_k$ is a vector of dimension $M^z_k \times 1$.  

As an example, consider the random vectors defined on the $5 \times 5$ lattice in Fig.~\ref{fig:tel}.  The random vector $z_0$ consists of the boundary points of the original $5\times 5$ lattice, $z_1$ is the boundary points left after removing $z_0$, and $z_2$ is the boundary point left after removing both $z_0$ and $z_1$.  The telescoping nature of $z_0$, $z_1$, and $z_2$ is clear since we start by defining $z_0$ on the boundary and telescope inwards to define subsequent random vectors.  The clockwise ordering of $z_k$ is shown by the arrows in Fig.~\ref{fig:tel}.
The telescoping recursive representation for $\bx$ is given in the following theorem.
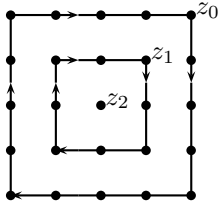
\begin{figure}
\begin{center}
\begin{pspicture}[unit=0.6cm](-3,-1.8)(3,1.6)
\psdots(-2,2)(-1,2)(0,2)(1,2)(2,2)
(-2,1)(-1,1)(0,1)(1,1)(2,1)
(-2,0)(-1,0)(0,0)(1,0)(2,0)
(-2,-1)(-1,-1)(0,-1)(1,-1)(2,-1)
(-2,-2)(-1,-2)(0,-2)(1,-2)(2,-2)
\psline{->}(-2,2)(-0.5,2)
\psline(-0.5,2)(2,2)
\psline{->}(2,2)(2,0.5)
\psline(2,0.5)(2,-2)
\psline{->}(-0.5,-2)(-2,-2)
\psline(-0.5,-2)(2,-2)
\psline{->}(-2,-2)(-2,0.5)
\psline(-2,0.5)(-2,2)
\psline{->}(-1,1)(-0.5,1)
\psline(-0.5,1)(1,1)
\psline{->}(1,1)(1,0.5)
\psline(1,0.5)(1,-1)
\psline{->}(-0.5,-1)(-1,-1)
\psline(-0.5,-1)(1,-1)
\psline{->}(-1,-1)(-1,0.5)
\psline(-1,0.5)(-1,1)
\rput[l](2.1,2.1){$z_0$}
\rput[l](1.1,1.1){$z_1$}
\rput[l](0.1,0.1){$z_2$}
\end{pspicture}
\caption{Telescoping recursions}
\label{fig:tel}
\end{center}
\end{figure}

\begin{theorem}[Telescoping Representation for GMRFs]
\label{thm:telescoping_discrete}
For a GMRF $x(i,j)$, the process $z_k$ defined in (\ref{eq:def_z_k}) is a Gauss-Markov process and thus admits a recursive representation
\begin{equation} 
z_k = F_k z_{k-1} + w_k \,, k = 1,2,\ldots, \tau \,, z_0 = \bx_b \label{eq:z_k}
\end{equation}
where $F_k$ is an $M^z_k \times M^z_{k-1}$ matrix and $w_k$ is white Gaussian noise independent of $z_{k-1}$ such that
\begin{align}
F_k &= E[z_k z_{k-1}^T] \left( E\left[z_{k-1} z_{k-1}^T\right]\right)^{-1} \label{eq:f_k_discrete} \\
Q_k &= E[w_k w_k^T] = E[z_k z_k^T] - F_k E[z_{k-1} z_k^T] \label{eq:q_k_discrete}\,.
\end{align}
\end{theorem}
\begin{proof}
The fact that $z_k$ is a Markov process follows from the global Markov property outlined in Theorem \ref{thm:global_markov_property}.  The recursive representation follows from standard theory of state-space models \cite{KailathSayed}.
\end{proof}

Both (\ref{eq:z_k}) and the continuous index telescoping representation are similar since the recursion initiates at the boundary and telescopes inwards.  For the continuous indices, the recursions were not unique, whereas for the discrete index case, the recursions are unique.  We outline a fast algorithm for computing $F_k$ and $Q_k$ that does not require knowledge of the covariance of $\bx$, just knowledge of the matrices ${\cal A}$ and ${\cal A}_b$ in (\ref{eq:matrix}).

Define the $NM \times 1$ vector $\bz$ such that (in Matlab$^\copyright$ notation)
\begin{equation}
\bz = \left[z_1 ; z_2 ; \ldots ; z_{\tau} \right] \,.
\end{equation}
The random vector $\bz$ is a permutation of the elements in $\bx$ such that
\begin{equation}
\bz = P \bx \,,
\end{equation}
where $P$ is a permutation matrix, which we know is orthogonal.  We can now write (\ref{eq:matrix}) in terms of $\bz$ by writing $\bx = P^T P \bx = P^T \bz$:
\begin{equation}
{\cal A}P^T \bz = {\cal A}_b z_0 + \bv \,. \label{eq:int_z_1}
\end{equation}
Multiplying both sides of (\ref{eq:int_z_1}) by $P$, we have
\begin{equation}
P {\cal A}P^T \bz = P{\cal A}_b z_0 + P\bv \,. \label{eq:int_z_2}
\end{equation}
Since $E[\bv \bv^T] = {\cal A}$, we have $E\left[(P\bv)(P\bv)^T\right] = P {\cal A} P^T$.  This suggests that (\ref{eq:int_z_2}) is a matrix based representation for the Gauss-Markov process $z_k$.  Further, because of the form ${\cal A}$ and ${\cal A}_b$, $P{\cal A} P^T$ and $P{\cal A}_b$ will have the form:

\bigskip

\noindent
$P{\cal A}P^T$
{\small
\begin{align}
&= \begin{bmatrix}
{\cal M}^0_1 & -{\cal M}^+_{1} & 0 & 0 & 0 \cdots & 0 \\
-{\cal M}^-_2 & {\cal M}^0_2 & {\cal M}^+_2 & 0 & \cdots & 0 \\
\vdots & \vdots & \vdots & \vdots & \vdots & \vdots \\
0 & \cdots & 0 & -{\cal M}^-_{{\cal K}-1} & {\cal M}^0_{{\cal K}-1} & -{\cal M}^+_{{\cal K}-1} \\
0 & \cdots & \cdots & 0 & -{\cal M}^-_{{\cal K}} & {\cal M}^0_{{\cal K}} \,.
\end{bmatrix} \label{eq:papt}\\
P {\cal A}_b &= [ {\cal M}^-_{1} \, ; \,0 ; \cdots \,; 0] \,,
\end{align}
}
where ${\cal M}_k^-$ is an $M^z_{k} \times M^z_{k-1}$ matrix, ${\cal M}_k^0$ is an $M^z_{k} \times M^z_{k}$ matrix, and ${\cal M}_k^+$ is an $M^z_{k} \times M^z_{k+1}$ matrix.  From (\ref{eq:vv_A}), $P {\cal A} P^T$ is positive and symmetric, and thus 
${\cal M}^+_{k} = \left({\cal M}^-_{k+1}\right)^T$.  To find the telescoping representation using (\ref{eq:int_z_2}), we find the Cholesky factors for $P{\cal A} P^T$ such that
\begin{align}
P{\cal A} P^T &= {\cal L}^T {\cal L} \label{eq:cholesky_p} \\
{\cal L} &= \begin{bmatrix} 
{\cal L}_{1} & 0 & 0 & \cdot & \cdot & 0 \\  
-{\cal P}_{2} & {\cal L}_{2} & 0 & 0 & \cdot & 0 \\
\cdot & \cdot & \cdot & \cdot & \cdot & \cdot \\
\cdot & \cdot & \cdot & \cdot & \cdot & \cdot \\
0 & \cdot & 0 & -{\cal P}_{\tau-1} & {\cal L}_{\tau-1} & 0 \\
0 & \cdot & \cdot & 0 & -{\cal P}_{\tau} & {\cal L}_{\tau}
\end{bmatrix} \,, \label{eq:chol_2m}
\end{align}
where the blocks ${\cal L}_k$ are $M^z_k \times M^z_k$ lower triangular matrices, and the blocks ${\cal P}_k$ are $M^z_k \times M^z_{k-1}$ matrices.  Substituting (\ref{eq:cholesky_p}) in (\ref{eq:int_z_2}) and inverting ${\cal L}^T$, we have
\begin{equation}
{\cal L} \bz = {\cal L}^{-T} P {\cal A}_b z_0 + {\cal L}^{-T} P \vec{v} \,.
\label{eq:recursive_chol}
\end{equation}
Notice that the noise is now white Gaussian since
\[ E\left[ {\cal L}^{-T} P \bv \bv^T P^T {\cal L}^{-1}\right] =
{\cal L}^{-T} P {\cal A} P^T {\cal L}^{-1} = I \,. \]
If we let ${\cal P}_1 = {\cal M}^{-}_{1}$, we can rewrite (\ref{eq:recursive_chol}) in recursive manner as
\begin{equation}
z_k = {\cal L}^{-1}_k {\cal P}_k z_{k-1} + w_k \,, \; k = 1,2,\ldots,\tau \,,
\end{equation}
where $F_k = {\cal L}^{-1}_k {\cal P}_k$ and $Q_k = {\cal L}_k^{-1} {\cal L}_k^{-T}$.  A recursive algorithm for calculating $F_k$ and $Q_k$, which follows from the calculation of the Cholesky factors, is given as follows \cite{MouraBalram1992}:

Initialization: $Q_{\tau} = ({\cal M}_{\tau}^0)^{-1}$, $F_{\tau} = Q_{\tau} {\cal M}_{\tau}^{-}$

For $k = \tau-1, \tau - 2, \ldots, 1$

\hspace{1cm} $Q_k^{-1} = {\cal M}^0_k - {\cal M}^+_k F_{k+1}$

\hspace{1cm} $F_k = Q_k {\cal M}^-_k$

end

\emph{Remark:}  The telescoping representations we derived shows the causal structure of Gauss-Markov random fields indexed over both continuous and discrete domains.  Our main result shows the existence of a recursive representation for GMRFs on telescoping surfaces that initiate at the boundary of the field and recurse inwards towards the center of the field.  Just like we derived estimation equations for GMRFs with continuous indices, we can use the telescoping representation to derive recursive estimation equations for GMRFs with discrete indices.  The numerical complexity of estimation will depend on the size of the state with maximum size, which for the GMRF is the perimeter of the field captured in the state $z_0$.  For example, the telescoping representation of a GMRF defined on a $\sqrt{N} \times \sqrt{N}$ lattice with non-zero boundary conditions will have a state of maximum size of order $O(\sqrt{N})$.  Notice that for both continuous and discrete indexed GMRFs, the telescoping representation is not local, \emph{i.e.}, each point in the GMRF does not depend on its neighborhood, but depends on the field values defined on a neighboring telescoping surface (or $F_k$ is not necessarily sparse).  Direct or straightforward implementation of the Kalman filter requires $O((\sqrt{N})^{3/2})$ due to a matrix inversion step.  However, using fast algorithms and appropriate approximations, fast implementation of Kalman filters, see \cite{SayedKailathLevAri1994} for an example, can lead to $O((\sqrt{N})^2)$, i.e., $O(N)$.

Now suppose the observations of the GMRF are given by
$
\by = H \bx + \bv\,,
$
where $\bx \in \R^n$ is the GMRF, $H$ is a diagonal matrix, and $\bv$ is white Gaussian noise vector such that $\bv \sim {\cal N}(0,R)$, where $R$ is a diagonal matrix. The mmse $\widehat{\bx}$ is a solution to the linear system,
\begin{equation}
[\Sigma^{-1} + H^{-1} R^{-1} H] \widehat{x} = H^T R^{-1} \by \,.
\label{eq:estimation_mmse}
\end{equation}
Since $\bx$ is a GMRF, it follows from \cite{SpeedKiiveri1986} that $\Sigma^{-1}$ is sparse, where the non-zero entries in $\Sigma^{-1}$ correspond to the edges in the graph\footnote{We note that the graphical models considered in this paper are a mixture of undirected and directed graphs, where the boundary values connect to nodes in a directed manner.  These graphs are examples of chain graphs, see \cite{Frydenberg1990}, and the underlying undirected graph can be recovered by moralizing this graph, i.e., converting directed edges into undirected edges and connecting edges between all boundary nodes.} associated with $\bx$.  In \cite{VatsMouraGraphJournal2010}, we use the telescoping representation to derive an iterative algorithms for solving (\ref{eq:estimation_mmse}) using the telescoping representation\footnote{The work in \cite{VatsMouraGraphJournal2010} applies to arbitrary graphical models and the telescoping representations are referred to as block-tree graphs.}.  Experimental results in \cite{VatsMouraGraphJournal2010} suggest that the numerical complexity of the iterative algorithm is $O(N)$, although the exact complexity may vary depending on the graphical model.  The use of the telescoping representation in deriving the iterative algorithm in \cite{VatsMouraGraphJournal2010}, is to identify computationally tractable local structures using the non-local telescoping representation.

\section{Summary}
\label{sec:summary}

We derived a recursive representation for noncausal Gauss-Markov random fields (GMRFs) indexed over regions in $\R^d$ or $\Z^d$, $d \ge 2$.  We called the recursive representation \emph{telescoping} since it initiated at the boundary of the field and telescoped inwards.  Although the equations for the continuous index case were derived assuming $x(t)$ is scalar, we can easily generalize the results for $x(t) \in \R^n$, $n \ge 1$.  Our recursions are on hypersurfaces in $\R^d$, which we call telescoping surfaces.  For fields indexed over $\R^d$, we saw that the set of telescoping surfaces is not unique and can be represented using a homotopy from the boundary of the field to a point within the field (not on the boundary).  Using the telescoping representations, we were able to recover recursive algorithms for recursive filtering and smoothing.  An extension of these results to random fields with two boundaries is derived in \cite{VatsMouraCDC2010}.  Besides the RTS smoother that we derived, other recursive smoothers can be derived using the results in \cite{KailathSayed,BadawiLP1979,FerrantePicci2000}.  We presented results for deriving recursive representations for GMRFs on lattices.  An example of applying this to image enhancement of noisy images is shown in \cite{VatsMoura2010icassp}.  Extensions of the telescoping representation to arbitrary graphical models are presented in \cite{VatsMouraGraphJournal2010}.  Using the results in \cite{VatsMouraGraphJournal2010}, we can derive computationally tractable estimation algorithms.

We note although the results derived in this paper assumed Gaussianity, recursive representations on telescoping surfaces can be derived for general non-Gaussian Markov random fields.  In this case, the representation will no longer be given by linear stochastic differential equation, but instead be transition probabilities.

\appendices

\section{Computing $b_j(s,r)$ in Theorem~\ref{thm:prediction}}
\label{app:example_gmrf}

We show how the coefficients $b_j(s,r)$ are computed for a GMRF $x(t)$, $t \in T \subset \R^d$.  Let $G_{-}$ and $G_{+}$ be complementary sets in $T \subset \R^d$ as shown in Fig.~\ref{fig:notation}.  Following \cite{Ldpitt1971}, define a function $h_s(t)$ such that
\begin{equation}
h_s(t) = \left\{ 
\begin{array}{cc}
R(s,t) & t \in G_{-} \cup \partial G \\
h_s(t) : {\cal L}_t h_s(t) = 0\,, & \\
h_s(\partial G) = R(s,\partial G) & t \in G_{+} \\
h_s(t) \in H_0^{m}(T) & 
\end{array}
\right. \,, \label{eq:a1_a}
\end{equation}
where ${\cal L}_r$ is defined in (\ref{eq:l_t}) and 
$H_0^{m}(T)$ is the completion of $C_0^{\infty}(T)$, the set of infinitely differentiable function with compact support in $T$, under the norm Sobolov norm order $m$.  From (\ref{eq:lr=d}), it is clear that $h_s(r) \ne R(s,r)$ when $r \in G_{+}$.  Let $u(r) \in C_0^{\infty}(T)$ and consider the following steps for computing $b_j(s,r)$:
\bigskip

\noindent
$\displaystyle{\sum_{|\alpha|,|\beta|\le m} \int_T D^{\alpha} u(t) a_{\alpha\beta}(t)D^{\beta} h_s(t) dt} $
\begin{align}
&= \sum_{|\alpha|,|\beta| \le m} \int_{G_{-}}
D^{\alpha} u(t) a_{\alpha,\beta}(t) D^{\beta} R(s,t) dr \nonumber \\
&\hspace{1cm}+ \sum_{|\alpha|,|\beta| \le m} \int_{G_{+}}
D^{\alpha} u(t) a_{\alpha,\beta}(t) D^{\beta} h_s(t) dr \label{eq:a1_2}\\
&= \sum_{j=0}^{m-1} \int_{\partial G}  b_{j}(s,r) \frac{\partial^j}{\partial n^j} u(r) d l + \int_{G_{-}} u(r) L_{t} R(s,t) dt \nonumber\\
&\hspace{1cm} + \int_{G_{+}} u(t) L_{t} h_s(t) dt \label{eq:a1_3}\\
&= \sum_{j=0}^{m-1} \int_{\partial G}  b_{j}(s,r) \frac{\partial^j}{\partial n^j} u(r) d l \,. \label{eq:a1_4}
\end{align}

To get (\ref{eq:a1_2}), we split the integral integral on the left hand side over $G_{-}$ and $G_{+}$.  In going from (\ref{eq:a1_2}) to (\ref{eq:a1_3}), we use integration by parts and the fact that $u(r) \in C_0^{\infty}(T)$.  We get (\ref{eq:a1_4}) using (\ref{eq:lr=d}) and (\ref{eq:a1_a}).  Thus, to compute $b_j(s,r)$, we first need to find $h_s(r)$ using (\ref{eq:a1_a}) and then use the steps in (\ref{eq:a1_2})$-$(\ref{eq:a1_4}).  We now present an example where we compute $b_j(s,r)$ for a Gauss-Markov process.

\textbf{Example:}
Let $x(t) \in \R$ be the Brownian bridge on $T = [0,1]$ such that
\begin{equation}
x(t) = w(t) - tw(1) \label{eq:bb}\,,
\end{equation}
where $w(t)$ is a standard Brownian motion.  Since covariance of $w(t)$ is $\min(t,s)$, the covariance of $x(t)$ is given by
\begin{equation}
R(t,s) = \left\{ \begin{array}{cc}
s(1-t) & t > s\\
t(1-s) & t < s
\end{array} \right. \,.
\end{equation}
Using the theory of reciprocal processes, see \cite{Krener1991,Levy2008}, it can be shown that the operator ${\cal L}_t$ is
\begin{equation}
L_t R(t,s) = -\frac{\partial^2 R(t,s)}{\partial t^2} = \delta(t-s) \,.
\end{equation}
Thus, the inner product associated with $R(t,s)$ is given by
\begin{equation}
<u,v> = \left\langle D u, Dv\right\rangle_T = 
\int_0^1 \frac{\partial}{\partial s} u(s)
\frac{\partial}{\partial s} v(s) ds \,.
\end{equation}
Following (\ref{eq:a1_a}), for $r < 1$ and $s \in [r,1]$, 
$h_s(t) = R(s,t)$ for $t \in [0,r]$ and
\begin{align}
-\frac{\partial^2 }{\partial t^2} h_s(t) &= 0 \,,\quad  t > r \\
h_s(r) &= r(1-s) \,,\quad h_s(1) = 0 \,.
\end{align}
We can trivially show that $h_s(t)$ is given by
\begin{equation}
h_s(t) = \frac{r(1-s)}{1-r} (1-t) \,, \quad t \ge r \,.
\end{equation}
We now follow the steps in (\ref{eq:a1_2})-(\ref{eq:a1_4}):
\begin{align*}
\int_0^1 \frac{\partial}{\partial t} u(t) \frac{\partial}{\partial t} h_s(t) dt
&= \int_0^r \frac{\partial}{\partial t} u(t) \frac{\partial}{\partial t} R(s,t) dt \\
&\hspace{1cm}+\int_r^1 \frac{\partial}{\partial t} u(t) \frac{\partial}{\partial t} h_s(t) dt \\
&= \left.u(t)\frac{\partial}{\partial t} R(s,t)\right|_0^r 
+  \left.u(t)\frac{\partial}{\partial t} h_s(t)\right|_r^1 \\
&= \left(\frac{\partial}{\partial t} R(s,r) - \frac{\partial}{\partial r} h_s(r)\right) u(r) \,. \\
&= \frac{1-s}{1-r} u(r) \,.
\end{align*}
Using Theorem \ref{thm:prediction}, we can compute $E[x(s)| \sigma\{x(t):0\le t \le r\}]$ as
\begin{equation}
E[x(s)| \sigma\{x(t):0\le t \le r\}] = E[x(s) | x(r)] = 
\left(\frac{1-s}{1-r}\right) x(r) \,. \label{eq:a1_5_e}
\end{equation}
We note that since $x(t)$ is a Gauss-Markov process, it is known that, \cite{Wong1985},
\begin{equation} 
E[x(s) | x(r)] = R(s,r) R^{-1}(r,r) x(r) \,. \label{eq:a1_5}
\end{equation}
Using the expression for $R(s,r)$, we can easily verify that (\ref{eq:a1_5_e}) and (\ref{eq:a1_5}) are equivalent.

\section{Proof of Theorem~\ref{thm:telescoping_rec_gmrf_unit_disc}: Telescoping Representation}
Let $\widehat{x}_{\lambda+d\lambda|\lambda}(\theta)$ denote the conditional expectation of $x_{\lambda+d\lambda}(\theta)$ given the $\sigma$-algebra generated by the field $\{x_{\mu}(\alpha):(\mu,\alpha) \in [0,\lambda]\times \Theta \}$.  From Theorem \ref{thm:prediction}, we have
\begin{equation}
\widehat{x}_{\lambda+d\lambda|\lambda}(\theta) = 
\sum_{j=0}^{m-1} \int_{\Theta}
b_j((\lambda+d\lambda,\theta),(\lambda,\alpha))
\frac{d^j}{d n^j} x_{\lambda}(\alpha) d\alpha \,.
\label{eq:a_1}
\end{equation}
It is clear that 
$x_{\lambda}(\theta) = \widehat{x}_{\lambda|\lambda}(\theta)$.  Taking the limit in (\ref{eq:a_1}) as $d \lambda \rightarrow 0$, we have
\begin{equation}
x_{\lambda}(\theta) = \sum_{j=0}^{m-1} \int_{\Theta}
b_j((\lambda,\theta),(\lambda,\alpha))
\frac{d^j}{d n^j} x_{\lambda}(\alpha) d\alpha \,.
\label{eq:a_2}
\end{equation}
Define the error as $\xi_{\lambda+d\lambda}(\theta)$ such that
\begin{equation}
\xi_{\lambda+d\lambda}(\theta) = x_{\lambda+d\lambda}(\theta) - \widehat{x}_{\lambda+d\lambda}(\theta) \,. \label{eq:a_3}
\end{equation}
Adding and subtracting $x_{\lambda}(\theta)$ in (\ref{eq:a_3}) and using (\ref{eq:a_2}), we have
\begin{align}
x_{\lambda+d\lambda}(\theta) - x_{\lambda}(\theta)
&= \sum_{j=0}^{m-1} \int_{\Theta}
[b_j((\lambda+d\lambda,\theta),(\lambda,\alpha)) \nonumber \\
&\hspace{-2cm} - b_j((\lambda,\theta),(\lambda,\alpha))]
\frac{d^j}{d n^j} x_{\lambda}(\alpha) d\alpha + \xi_{\lambda+d\lambda}(\theta) \,. \label{eq:aa_3}
\end{align}
Assuming $d\lambda$ is small, we can write $b_j((\lambda+d\lambda,\theta),(\lambda,\alpha))
- b_j((\lambda,\theta),(\lambda,\alpha))$ as

\bigskip

\noindent
$b_j((\lambda+d\lambda,\theta),(\lambda,\alpha))
- b_j((\lambda,\theta),(\lambda,\alpha))$
\begin{align}
& = 
\frac{b_j((\lambda+d\lambda,\theta),(\lambda,\alpha))
- b_j((\lambda,\theta),(\lambda,\alpha))}{d\lambda} d\lambda \label{eq:a_4}\\
&= \left(\lim_{\mu \rightarrow \lambda^{+}} \frac{\partial}{\partial \mu}
b_j((\mu,\theta),(\lambda,\alpha))\right) d\lambda \,, \label{eq:a_5}
\end{align}
where in going from (\ref{eq:a_4}) to (\ref{eq:a_5}), we use the assumption that $d\lambda$ is close to zero.  Writing $dx_{\lambda}(\theta) = x_{\lambda+d\lambda}(\theta) - x_{\lambda}(\theta)$ and substituting (\ref{eq:a_5}) in (\ref{eq:aa_3}), we get
\begin{equation}
dx_{\lambda}(\theta) = F_{\theta} [x_{\lambda}(\theta)] d\lambda + \xi_{\lambda+d\lambda}(\theta) \,,
\end{equation}
where $F_{\theta}$ is given in (\ref{eq:f_theta_unit_disc}).  To get the final form of the telescoping representation, we need to characterize $\xi_{\lambda+d\lambda}(\theta)$.  To do this, we write $\xi_{\lambda+d\lambda}(\theta)$ as
\begin{equation}
\xi_{\lambda+d\lambda}(\theta) = 
B_{\lambda}(\theta) d w_{\lambda}(\theta) = B_{\lambda}(\theta)[w_{\lambda+d\lambda}(\theta) - w_{\lambda}(\theta)] \,. \label{eq:xi_def_w}
\end{equation}
We now prove the properties of $w_{\lambda}(\theta)$:
\begin{enumerate}[i)]
\item Since $\xi_{\lambda}(\theta) = x_{\lambda}(\theta) - \widehat{x}_{\lambda|\lambda}(\theta)$ and $\widehat{x}_{\lambda|\lambda}(\theta) = x_{\lambda}(\theta)$ by definition, we have
\[ \lim_{d\lambda \rightarrow 0} \xi_{\lambda+d\lambda}(\theta) = \xi_{\lambda}(\theta) = 0\,, a.s. \,.\]
Thus, using (\ref{eq:xi_def_w}), since $B_{\lambda}(\theta) \ne 0$, we have
\begin{align}
\lim_{d \lambda \rightarrow 0} w_{\lambda+d\lambda}(\theta) - w_{\lambda}(\theta) &= 0 , a.s. \\
\lim_{d \lambda \rightarrow 0} w_{\lambda+d\lambda}(\theta) &= 
w_{\lambda}(\theta), a.s. \label{eq:k_1}
\end{align}
Equation (\ref{eq:k_1}) shows that $w_{\lambda}(\theta)$ is almost surely continuous in $\lambda$.

\item Since the driving noise at the boundary of the field can be captured in the boundary conditions, without loss in generality, we can assume that $w_0(\theta) = 0$ for all $\theta \in \Theta$.  

\item For $0 \le \lambda_1 \le \lambda_1' \le \lambda_2 \le \lambda_2'$ and $\theta_1,\theta_2 \in \Theta$, let $d\lambda_1 = \lambda_1' - \lambda_1$ and $d\lambda_2 = \lambda_2' - \lambda_2$. Consider the covariance

\bigskip

\noindent
$E[\xi_{\lambda_1+d\lambda_1}(\theta_1)
\xi_{\lambda_2+d\lambda_2}(\theta_2]$
\begin{align}
&=E\left[\left(x_{\lambda_1+d\lambda_1}(\theta_1) - 
\widehat{x}_{\lambda_1+d\lambda_1}(\theta_1)\right)
\xi_{\lambda_2+d\lambda_2}(\theta_2)\right] \label{eq:t_1}\\
&= E[x_{\lambda_1+d\lambda_1}(\theta_1) \xi_{\lambda_2+d\lambda_2}(\theta_2)]  \label{eq:t_2} \\
&\hspace{2cm} -E[\widehat{x}_{\lambda_1 + d\lambda_1}(\theta_1) \xi_{\lambda_2+d\lambda_2}(\theta_2)] \nonumber \\
&= 0 \,,
\end{align}
where to go from (\ref{eq:t_1}) to (\ref{eq:t_2}), we use the orthogonality of the error.  Using the definition of $\xi_{\lambda+d\lambda}(\theta)$ in (\ref{eq:xi_def_w}), we have that $w_{\lambda_1'}(\theta_1) - w_{\lambda_1}(\theta_1)$ and $w_{\lambda_2'}(\theta_2) - w_{\lambda_2}(\theta_2)$ are independent random variables.

\item We now compute $E[\xi_{\lambda+d\lambda}(\theta_1)
\xi_{\lambda+d\lambda}(\theta_2)]$:

\noindent
$E[\xi_{\lambda+d\lambda}(\theta_1)\xi_{\lambda+d\lambda}(\theta_2)]$
\begin{align}
&=E[\xi_{\lambda+d\lambda}(\theta_1)(x_{\lambda+d\lambda}(\theta_2) - \widehat{x}_{\lambda+d\lambda|\lambda}(\theta_2))] \nonumber \\
&= E[(x_{\lambda+d\lambda}(\theta_1) - \widehat{x}_{\lambda+d\lambda|\lambda}(\theta_1))x_{\lambda+d\lambda}(\theta_2)] \nonumber \\
&= R_{\lambda+d\lambda,\lambda+d\lambda}(\theta_1,\theta_2)\nonumber\\
&-\sum_{j=0}^{m-1} \!\!\int_{\Theta}\!\! b_j((\lambda+d\lambda,\theta_1),(\lambda,\alpha))
\frac{\partial^j}{\partial n^j} R_{\lambda,\lambda+d\lambda}(\alpha,\theta_2) d\alpha \nonumber \\
&= R_{\lambda+d\lambda,\lambda+d\lambda}(\theta_1,\theta_2)
- R_{\lambda,\lambda+d\lambda}(\theta_1,\theta_2) \nonumber \\
&\hspace{1cm} + R_{\lambda,\lambda+d\lambda}(\theta_1,\theta_2) \label{eq:last_1}\\
&- \sum_{j=0}^{m-1} \!\!\int_{\Theta} \!\! b_j((\lambda+d\lambda,\theta_1),(\lambda,\alpha))
\frac{\partial^j}{\partial n^j}R_{\lambda,\lambda+d\lambda}(\alpha,\theta_2) d\alpha \,.
\nonumber 
\end{align} 
Using (\ref{eq:a_2}), we have

\noindent
$R_{\lambda,\lambda+d\lambda}(\theta_1,\theta_2)$
\begin{align}
&=
E[x_{\lambda}(\theta_1) x_{\lambda+d\lambda}(\theta_2)] \\
&= \sum_{j=0}^{m-1} \int_{\Theta}
b_j((\lambda,\theta_1),(\lambda,\alpha))
\frac{\partial^j}{\partial n^j}R_{\lambda,\lambda+d\lambda}(\alpha,\theta_2) d\alpha \,.
\label{eq:ff_1}
\end{align}
Substituting (\ref{eq:ff_1}) in (\ref{eq:last_1}), we have
\smallskip

{\small 
\noindent
$E[\xi_{\lambda+d\lambda}(\theta_1)\xi_{\lambda+d\lambda}(\theta_2)]$
\begin{align}
&= R_{\lambda+d\lambda,\lambda+d\lambda}(\theta_1,\theta_2)
- R_{\lambda,\lambda+d\lambda}(\theta_1,\theta_2)  \\
&\hspace{0.6cm} - \sum_{j=0}^{m-1} \int_{\Theta} [b_j((\lambda+d\lambda,\theta_1),(\lambda,\alpha)) \\
&\hspace{0.6cm} - b_j((\lambda,\theta_1),(\lambda,\alpha))]
\frac{\partial^j}{\partial n^j}R_{\lambda,\lambda+d\lambda}(\alpha,\theta_2) d\alpha \nonumber \\
&= \left[\frac{R_{\lambda+d\lambda,\lambda+d\lambda}(\theta_1,\theta_2)
- R_{\lambda,\lambda+d\lambda}(\theta_1,\theta_2)}{d\lambda}\right] d\lambda \nonumber \\
&- \left[\sum_{j=0}^m \int_{\Theta} \left(\frac{b_j((\lambda+d\lambda,\theta_1),(\lambda,\alpha)) - 
b_j((\lambda,\theta_1),(\lambda,\alpha))}{d\lambda}\right) \right. \nonumber \\
& \hspace{3.3cm} \left.\frac{\partial^j}{\partial n^j}R_{\lambda,\lambda+d\lambda}(\alpha,\theta_2) d\alpha \right]
d\lambda \nonumber \\
&= \left(\lim_{\mu \rightarrow \lambda^-} \frac{\partial}{\partial \mu} R_{\mu,\lambda}(\theta_1,\theta_2) \nonumber \right.\\
&\left.- \lim_{\mu \rightarrow \lambda^+} \frac{\partial}{\partial \mu}
\sum_{j=0}^{m-1} \int_{\Theta} b^j_{\mu,\lambda}(\theta_1,\alpha)
\frac{\partial^j}{\partial n^j} R_{\lambda,\lambda}(\alpha,\theta_2) d\alpha  \right) d \lambda \nonumber \\
&= \left(\lim_{\mu \rightarrow \lambda^-} \frac{\partial}{\partial \mu} R_{\mu,\lambda}(\theta_1,\theta_2) - \lim_{\mu \rightarrow \lambda^+} \frac{\partial}{\partial \mu} R_{\mu,\lambda}(\theta_1,\theta_2) \right) d \lambda \nonumber \\
&= C_{\lambda}(\theta_1,\theta_2) d\lambda \,. \label{eq:c_l_theta_theta}
\end{align}
}
Thus, for $d\lambda$ small, we have
\begin{align}
E\left[(w_{\lambda+d\lambda}(\theta_1) - w_{\lambda}(\theta_1))
(w_{\lambda+d\lambda}(\theta_2) - w_{\lambda}(\theta_2))\right] &\nonumber \\ 
&\hspace{-4cm}=\frac{C_{\lambda}(\theta_1,\theta_2)}{B_{\lambda}(\theta_1)B_{\lambda}(\theta_2)} d\lambda \,. \label{eq:ttt_1}
\end{align}

Since $w_0(\theta) = 0$, we can use (\ref{eq:ttt_1}) to compute $E[w_{\lambda}(\theta_1) w_{\lambda}(\theta_2)]$ as follows:

\bigskip
$E[(w_{\lambda}(\theta_1) - w_0(\theta_1))
(w_{\lambda}(\theta_2) - w_0(\theta_2))]$
\begin{align}
&= \lim_{N\rightarrow \infty}E\left[ \sum_{k = 0}^{N+1} (w_{\gamma_k}(\theta_1) - w_{\gamma_{k-1}}(\theta_2)) \right. \nonumber \\
&\left.\sum_{k = 0}^{N+1} (w_{\gamma_k}(\theta_2) - w_{\gamma_{k-1}}(\theta_2))
\right], \gamma_0 = \lambda, \gamma_{N+1} = 0 \hspace{-0.05cm}\label{eq:dd_1}\\
&= \lim_{N\rightarrow \infty} E\left[\sum_{k=0}^{N+1}
(w_{\gamma_k}(\theta_1) - w_{\gamma_{k-1}}(\theta_1)) \right. \nonumber \\
&\hspace{1cm}\left.\phantom{\sum_{1}^{2}}(w_{\gamma_k}(\theta_2) - w_{\gamma_{k-1}}(\theta_2))
\right] \label{eq:dd_2}\\
&= \lim_{N\rightarrow \infty} \sum_{k=0}^{N+1} \frac{C_{\gamma_{k-1}}(\theta_1,\theta_2)}
{B_{\gamma_{k-1}(\theta_1)B_{\gamma_{k-1}}(\theta_2)}}(\gamma_k - \gamma_{k-1}) \label{eq:dd_3}\\
&= \int_0^{\lambda} \frac{C_u(\theta_1,\theta_2)}
{B_u(\theta_1) B_u(\theta_2)} du \,. \label{eq:dd_4}
\end{align}

We get (\ref{eq:dd_2}) using the orthogonal increments property in (iii).  We use (\ref{eq:ttt_1}) to get (\ref{eq:dd_3}).  We use the definition of the Riemann integrals to go from (\ref{eq:dd_3}) to (\ref{eq:dd_4}).

\item For $\lambda_1 > \lambda_2$, the covariance of $w_{\lambda_1}(\theta) - w_{\lambda_2}(\theta)$ is computed as follows:

\smallskip

\noindent
$E[(w_{\lambda_1}(\theta) - w_{\lambda_2}(\theta))^2]$
\begin{align}
&=E[w_{\lambda_1}^2(\theta)] + E[w_{\lambda_2}^2(\theta)]
-2 E[w_{\lambda_1}(\theta) w_{\lambda_2}(\theta)] \\
&= \lambda_1 + \lambda_2 - 2  \int_0^{\lambda_2} \frac{C_u(\theta_1,\theta_2)}
{B_u(\theta_1) B_u(\theta_2)} \,du \,,
\end{align}
where
\begin{align}
E[w_{\lambda}^2(\theta)] &= \int_0^{\lambda} \frac{C_u(\theta,\theta)}{B_u^2(\theta)} du \label{eq:cov_1}\\
&= \int_{\{u\in[0,\lambda]: C_u(\theta,\theta) = 0\}} \frac{C_u(\theta,\theta)}{B_u^2(\theta)} du \nonumber \\ 
&+\int_{\{u\in[0,\lambda]: C_u(\theta,\theta) \ne 0\}} \frac{C_u(\theta,\theta)}{B_u^2(\theta)} du \label{eq:cov_2}\\
&= \int_{\{u\in[0,\lambda]: C_u(\theta,\theta) \ne 0\}} du \label{eq:cov_3}\\
&= \lambda \label{eq:cov_4}\,.
\end{align}
To go from (\ref{eq:cov_2}) to (\ref{eq:cov_3}), we use (\ref{eq:b_l_unit_disc}) since $C_u(\theta,\theta) \ne 0$.  To go from (\ref{eq:cov_3}) to (\ref{eq:cov_4}), we use the given assumption that the set $\{u\in[0,1]: C_u(\theta,\theta) = 0\}$ has measure zero.
\end{enumerate}

\section{Proof of Theorem~\ref{thm:recursive_filter}: Recursive Filter}
\label{appendix_proof_filter}
The steps involved in deriving the recursive filter are the same as deriving the Kalman-Bucy filtering equations, see \cite{Oksendal2005}.  The only difference is that we need to take into account the dependence of each point in the random field on its neighboring telescoping surface (which is captured in the integral transform $F_{\theta}$), instead of a neighboring point as we do for Gauss-Markov processes.  The steps in deriving the recursive filter are summarized as follows.  In Step 1, we define the innovation process and show that it is Brownian motion and equivalent to the observation space.  Using this, we find a relationship between the filtered estimate and the innovation, see Lemma \ref{lemma:innovation}.  In Step 2, we find a representation for the field $x_{\lambda}(\theta)$, see Lemma \ref{lemma:solution}.  Using Lemma \ref{lemma:innovation} and Lemma \ref{lemma:solution}, we find a closed form expression for $\widehat{x}_{\lambda | \lambda}(\theta)$ in Step 3.  We differentiate this to derive the equation for the filtered estimate in Step 4.  Finally, Step 5 computes the equation for the error covariance.
\bigskip

\noindent
\textbf{Step 1. [Innovations]}
Define $q_{\lambda}(\theta)$ such that
\[ q_{\lambda}(\theta) = y_{\lambda}(\theta) - 
\int_0^{\lambda} G_{\mu}(\theta) \widehat{x}_{\mu |\mu}(\theta) d\mu \,. \]
Define the innovation field $e_{\lambda}(\theta)$ such that
\begin{align}
de_{\lambda}(\theta) &= \frac{1}{D_{\lambda}(\theta)} \, dv_{\lambda}(\theta) \\
&= \frac{G_{\lambda}(\theta)}{D_{\lambda}(\theta)} \, \widetilde{x}_{\lambda|\lambda}(\theta) d \lambda + d n_{\lambda}(\theta) \,,
\label{eq:innovation}
\end{align}
where we have used (\ref{eq:observations}) to get the final expression in (\ref{eq:innovation}) and assume that $D_{\lambda}(\theta) \ne 0$.

\begin{lemma}
\label{lemma:innovation}
The field $e_{\lambda}(\theta)$ is Brownian motion for each fixed $\theta$ and $E[e_{\lambda_1}(\theta_1) e_{\lambda_2}(\theta_2)] = 0$ when $\lambda_1 \ne \lambda_2, \theta_1 \ne \theta_2$.
\end{lemma}
\begin{proof}
Note that $E[\widetilde{x}_{\lambda|\lambda}(\theta) | \sigma\{e_{\mu}(\theta):0\le\mu\le\lambda\}] = 0$ since $\widetilde{x}_{\lambda|\lambda}(\theta) \perp y_{\mu}(\alpha)$ for $\alpha \in \Theta$, $0 \le \mu \le \lambda$.  Thus, using Corollary 8.4.5 in \cite{Oksendal2005}, we establish that $e_{\lambda}(\theta)$ is Brownian motion for each fixed $\theta$.
Assume $\lambda_1 > \lambda_2$ and consider $\gamma < \lambda_2$.  
Then, using the orthogonality of error, $\widetilde{x}_{\mu}(\theta_1) \perp y_{\gamma}(\alpha)$ for $\gamma < \mu$, and the fact that $n_{\lambda}(\theta) \perp e_{\gamma}(\alpha)$ for $\gamma < \lambda$, we have

\smallskip

\noindent
$E[(e_{\lambda_1}(\theta_1) - e_{\lambda_2}(\theta_2))e_{\gamma}(\alpha)]$
\begin{align}
&=\int_{\lambda_2}^{\lambda_1} \left(
\frac{G_{\mu}(\theta)}{D_{\mu}(\theta)}\right) E[\widetilde{x}_{\mu|\mu}(\theta_1) e_{\gamma}(\alpha)] d\mu \nonumber \\
&\hspace{1cm} + E[(n_{\lambda_1}(\theta_1) -n_{\lambda_2}(\theta_1) )e_{\gamma}(\alpha)]  \\
&= 0 \,.
\end{align}
Now we compute $E[de_{\lambda_1}(\theta_1) de_{\lambda_2} (\theta_2)]$ for $\lambda_1 > \lambda_2$:

\smallskip

\noindent
$E[de_{\lambda_1}(\alpha_1) de_{\lambda_2} (\theta_2)]$
\begin{align}
&=E\left[\left( \frac{G_{\lambda_1}(\alpha_1)}{D_{\lambda_1}(\theta_1)} \widetilde{x}_{\lambda_1|\lambda_1}(\theta_1) d \lambda_1 + d n_{\lambda_1}(\theta_1) \right)
d e_{\lambda_2}(\theta_2)\right] \label{eq:pf_1}\\
&= E[ d n_{\lambda_1}(\theta_1) d e_{\lambda_2}(\theta_2)] \label{eq:pf_2} \\
&= \frac{1}{D_{\lambda_2}(\theta_2)}
E\left[d n_{\lambda_1}(\theta_1)\left( dy_{\lambda_2}(\theta_2) \nonumber \right. \right. \\
&\hspace{2cm} \left.\left.- G_{\lambda_2}(\theta_2) 
\widehat{x}_{\lambda_2 | \lambda_2}(\theta_2) d\lambda_2 \right)\right] \label{eq:pf_3} \\
&= \frac{1}{D_{\lambda_2}(\theta_2)} E[dn_{\lambda_1}(\theta_1)dy_{\lambda_2}(\theta_2)] \label{eq:pf_4}\\
&= \frac{1}{D_{\lambda_2}(\theta_2)} E[dn_{\lambda_1}(\theta_1)
\left( G_{\lambda_2}(\theta_2) x_{\lambda_2}(\theta_2) d\theta_2 \nonumber \right. \\
&\hspace{2cm} \left.+ D_{\lambda_2}(\theta_2) d n_{\lambda_2}(\theta_2) \right)] \label{eq:pf_5} \\
&= E[dn_{\lambda_1}(\theta_1) dn_{\lambda_2}(\theta_2)] = 0 \label{eq:pf_6} \,.
\end{align}
To go from (\ref{eq:pf_1}) to (\ref{eq:pf_2}), we use that $\widetilde{x}_{\lambda_1 | \lambda_1}(\theta_1)$ is independent of $d e_{\lambda_2}(\theta_2)$, since $de_{\lambda_2}(\theta_2)$ is a linear combination on the observations $\{y(\partial T^s), s \in [0,\lambda_2]\}$.  We get (\ref{eq:pf_3}) using the definition of $de_{\lambda_2}(\theta_2)$ in (\ref{eq:innovation}).  To go from (\ref{eq:pf_3}) to (\ref{eq:pf_4}), we use that $dn_{\mu_1}(\theta_1)$ is independent of $\widehat{x}_{\lambda_2 | \lambda_2}(\theta_2)$ for $\lambda_1 > \lambda_2$.  We get (\ref{eq:pf_5}) using the equation for the observations in (\ref{eq:observations}).  To go from (\ref{eq:pf_5}) to (\ref{eq:pf_6}), we use the assumption that $n_{\lambda_1}(\theta_1)$ is independent of the GMRF $x_{\lambda}(\theta)$.  In a similar manner, we can get the result for $\lambda_1 < \lambda_2$.

For $\lambda_1 \ne \lambda_2$ and $\theta_1 \ne \theta_2$, $E[e_{\lambda_1}(\theta_1) e_{\lambda_2}(\theta_2)] = 0$ follows from similar computations as done in (\ref{eq:dd_1})$-$(\ref{eq:dd_4}).
\end{proof}

Lemma~\ref{lemma:innovation} says that the innovation has the properties as the noise observation $n_{\lambda}(\theta)$.  We now use the innovation to find a closed form expression for the filtered estimate $\widehat{x}_{\lambda|\lambda}(\theta)$.

\begin{lemma}
\label{lemma:f_s}
The filtered estimate $\widehat{x}_{\lambda|\lambda}(\theta)$ can be written in terms of the innovation as
\begin{align}
\widehat{x}_{\lambda|\lambda}(\theta) &= \int_0^{\lambda} \int_{\Theta}
g_{\lambda,\mu}(\theta,\alpha)
d e_{\mu}(\alpha) d \alpha \label{eq:x_hat}\\
g_{\lambda,\mu}(\theta,\alpha) &= \frac{\partial}{\partial \mu} E[x_{\lambda}(\theta) e_{\mu}(\alpha)] \label{eq:g_l_u}\,.
\end{align}
\end{lemma}
\begin{IEEEproof}
Using the methods in \cite{Oksendal2005} or \cite{KailathSayed}, we can establish the equivalence between the innovations and the observations.  Because of this equivalence, we can write the filtered estimate as in (\ref{eq:x_hat}).  
We now compute $g_{\lambda,\mu}(\theta,\alpha)$.  We know that
\[ \left(x_{\lambda}(\theta) - \widehat{x}_{\lambda | \lambda}(\theta)\right)
\perp e_{\mu}(\alpha) \,,\quad \mu \le \lambda, \alpha \in \Theta \,. \]
Thus, we have

\smallskip

\noindent
$E[x_{\lambda}(\theta) e_{\mu}(\alpha)]$
\begin{align}
 &= E[\widehat{x}_{\lambda|\lambda}(\theta) e_{\mu}(\alpha)] \\
&= \int_0^{\lambda} \int_{\Theta} g_{\lambda,s}(\theta,\beta) E[d e_{s}(\beta) e_{\mu}(\alpha)] d \beta \\
&= \int_0^{\lambda} \int_{\Theta} \int_0^{\mu} 
g_{\lambda,s}(\theta,\beta) E[d e_s(\beta) d e_r(\alpha)] d \beta \label{eq:pff_1}\\
&= \int_0^{\lambda}\int_0^{\mu} \int_{\Theta} g_{\lambda,s}(\theta,\beta)
\delta(s-r) \delta(\beta-\alpha) ds dr d \beta \label{eq:pff_2}\\
&= \int_0^{\mu} g_{\lambda,r}(\theta,\alpha) dr \,. \label{eq:pff_3}
\end{align}
To go from (\ref{eq:pff_1}) to (\ref{eq:pff_2}), we use Lemma~\ref{lemma:innovation}.  Differentiating (\ref{eq:pff_3}) with respect to $\mu$, we get the expression for $g_{\lambda,\mu}(\theta,\alpha)$ in (\ref{eq:g_l_u}).
\end{IEEEproof}

\noindent
\textbf{Step 2. [Formula for $x_{\lambda}(\theta)$]}  
Before deriving a closed form expression for $x_{\lambda}(\theta)$, we first need the following Lemma.

\begin{lemma}
\label{thm:property_ft}
For any function $\Psi_{\lambda,\gamma}(\theta_1,\theta_2)$ with $m-1$ normal derivatives, we have
\begin{equation}
\int_0^{\lambda} F_{\theta_1} [\Psi_{\lambda,\gamma}(\theta_1,\theta_2)] d\gamma = 
F_{\theta_1} \left[\int_0^{\lambda}  \Psi_{\lambda,\gamma}(\theta_1,\theta_2) d\gamma \right] \,. \label{eq:property_ft}
\end{equation}
\end{lemma}
\begin{IEEEproof}
Using the definition of $F_{\theta_1}$, we have

\smallskip

\noindent
$\int_0^{\lambda} F_{\theta_1} [\Psi_{\lambda,\gamma}(\theta_1,\theta_2)] d\gamma$
\begin{align}
&=\int_0^{\lambda} \sum_{j=0}^{m-1} \int_{\Theta} \lim_{\mu \rightarrow \lambda^+} \frac{\partial}{\partial \mu}
b_j((\mu,\theta),(\lambda,\alpha)) \nonumber \\
&\hspace{2cm}  \lim_{h \rightarrow 0} \frac{\partial^j}{\partial h^j} 
\Psi_{\lambda + h \dot{\lambda},\gamma}(\alpha + h \dot{\alpha},\theta_2) d\alpha d\gamma \\
&=
\sum_{j=0}^{m-1} \int_{\Theta} \lim_{\mu \rightarrow \lambda^+} \frac{\partial}{\partial \mu}
b_j((\mu,\theta),(\lambda,\alpha)) \nonumber \\
& \hspace{1cm} \lim_{h \rightarrow 0} \frac{\partial^j}{\partial h^j} \left[\int_0^{\lambda} 
\Psi_{\lambda + h \dot{\lambda},\gamma}(\alpha + h \dot{\alpha},\theta_2) d\gamma\right] d\alpha \\
&= F_{\theta_1} \left[\int_0^{\lambda}  \Psi_{\lambda,\gamma}(\theta_1,\theta_2) d\gamma \right] \,. 
\end{align}
\end{IEEEproof}

\begin{lemma}
\label{lemma:solution}
Using the telescoping representation for $x_{\lambda}(\theta)$, a solution for $x_{\lambda}(\theta)$ is given as follows:
\begin{align}
x_{\lambda}(\theta) &= \int_{\Theta}
\Phi_{\lambda,\mu}(\theta,\alpha) x_{\mu} (\alpha) d\alpha \nonumber\\
&\hspace{2cm}+ \int_{\Theta} \int_{\mu}^{\lambda} \Phi_{\lambda,\gamma}(\theta,\alpha) 
d w_{\gamma}(\alpha) d\alpha \,,
\label{eq:x_l_l}
\end{align}
\[
\frac{\partial}{\partial \lambda} \Phi_{\lambda,\mu}(\theta,\alpha)
= F_{\theta} \left[\Phi_{\lambda,\mu}(\theta,\alpha)\right] \,,
\Phi_{\lambda,\lambda} = \delta(\theta - \alpha) \,.
\]
\end{lemma}
\begin{IEEEproof}
We show that (\ref{eq:x_l_l}) satisfies the differential equation in (\ref{eq:telescoping_arbit}).  Taking derivative of (\ref{eq:x_l_l}) with respect to $\lambda$, we have

\smallskip

\noindent
$d x_{\lambda}(\theta)$
\begin{align}
 &= \int_{\Theta} \frac{\partial}{d \lambda}
\Phi_{\lambda,\mu}(\theta,\alpha) x_{\mu}(\alpha) d\alpha d\lambda
+ \int_{\Theta} \Phi_{\lambda,\lambda}(\theta,\alpha) dw_{\lambda}(\alpha) \nonumber \\
&\hspace{2cm}+ \int_{\Theta} \int_{\mu}^{\lambda} \frac{\partial}{d \lambda}
\Phi_{\lambda,\gamma}(\theta,\alpha) dw_{\gamma}(\alpha) d \alpha \\
&= \int_{\Theta} F_{\theta} \left[\Phi_{\lambda,\mu}(\theta,\alpha)\right] x_{\mu}(\alpha)d\alpha d\lambda + dw_{\lambda}(\theta) \nonumber \\
&\hspace{2cm}+ \int_{\Theta} \int_{\mu}^{\lambda}F_{\theta}
\left[\Phi_{\lambda,\gamma}(\theta,\alpha)\right] dw_{\gamma}(\alpha) d\alpha \\
&= F_{\theta} \!\left[\! \int_{\Theta}\!
\Phi_{\lambda,\mu}(\theta,\alpha) x_{\mu} (\alpha) d\alpha
\!+ \!\!\int_{\Theta} \!\int_{\mu}^{\lambda}\! \Phi_{\lambda,\gamma}(\theta,\alpha) 
d b_{\gamma}(\alpha) d\alpha\right] \nonumber \\
& \hspace{6cm} + dw_{\lambda}(\theta) \label{eq:pfff_1}\\
&= F_{\theta} [x_{\lambda}(\theta)] + dw_{\lambda}(\theta) \,.
\end{align}
To get (\ref{eq:pfff_1}), we use Theorem \ref{thm:property_ft} to take the integral transform $F_{\theta}$ outside the integral. Since the $x_{\lambda}(\theta)$ in (\ref{eq:x_l_l}) satisfies (\ref{eq:telescoping_arbit}), it must be a solution.
\end{IEEEproof}

\noindent
\textbf{Step 3. [Equation for $\widehat{x}_{\lambda|\lambda}(\theta)$]}
Using (\ref{eq:innovation}) and (\ref{eq:x_l_l}), we can write $g_{\lambda,\mu}(\theta,\alpha)$ in (\ref{eq:g_l_u}) as

\noindent
$g_{\lambda,\mu}(\theta,\alpha)$
\begin{align}
 &= 
\frac{\partial}{\partial \mu} 
E\left\{x_{\lambda}(\theta) 
\left[ 
\int_0^{\mu} 
\frac{G_{\gamma}(\alpha)}{D_{\gamma}(\alpha)} \widetilde{x}_{\gamma|\gamma}(\alpha) d\gamma + n_{\mu}(\alpha)
\right]  \right\} \\
&= \frac{\partial}{\partial \mu} \left[ 
\int_0^u \frac{G_{\gamma}(\alpha)}{D_{\gamma}(\alpha)}
E[x_{\lambda}(\theta) \widetilde{x}_{\gamma| \gamma}(\alpha)] 
d\gamma \right] \\
&= \frac{\partial}{\partial \mu} 
\left\{
\int_0^u \frac{G_{\gamma}(\alpha)}{D_{\gamma}(\alpha)}
\int_{\Theta} \Phi_{\lambda,\gamma}(\theta,\beta) E[x_{\gamma}(\beta)\widetilde{x}_{\gamma| \gamma}(\alpha)] d \beta
d\gamma \right\} \\
&= \frac{G_{\mu}(\alpha)}{D_{\mu}(\alpha)}
\int_{\Theta} \Phi_{\lambda,\mu}(\theta,\beta) E[x_{\mu}(\beta)\widetilde{x}_{\mu | \mu}(\alpha)] d \beta \\
&= \frac{G_{\mu}(\alpha)}{D_{\mu}(\alpha)}
\int_{\Theta} \Phi_{\lambda,\mu}(\theta,\beta) 
S_{\mu}(\beta,\alpha)d \beta \,. \label{eq:g_l}
\end{align}
To get (\ref{eq:g_l}), we use the fact that 
$\widehat{x}_{\mu|\mu}(\alpha) \perp \widetilde{x}_{\mu| \mu}(\beta)$, so that
\[
E[x_{\mu}(\beta) \widetilde{x}_{\mu | \mu}(\alpha)]
= E[\widetilde{x}_{\mu|\mu}(\beta) \widetilde{x}_{\mu|\mu}(\alpha)] = S_{\mu}(\beta,\alpha) \,.
\]
Substituting (\ref{eq:g_l}) in the expression for $\widehat{x}_{\lambda|\lambda}(\theta)$ in (\ref{eq:g_l_u}) (Step~2), we get

\noindent
$\widehat{x}_{\lambda|\lambda}(\theta)$
\begin{equation}
 = 
\int_0^{\lambda} \int_{\Theta}
\frac{G_{\mu}(\alpha)}{D_{\mu}(\alpha)}
\left[\int_{\Theta} \Phi_{\lambda,\mu}(\theta,\beta) 
S_{\mu}(\beta,\alpha)d \beta \right]
d e_{\mu}(\alpha) d \alpha \label{eq:x_h_l} \,.
\end{equation}

\noindent
\textbf{Step 4. [Differential Equation for $\widehat{x}_{\lambda|\lambda}(\theta)$]}
Differentiating (\ref{eq:x_h_l}) with respect to $\lambda$, we get

\noindent
$d \widehat{x}_{\lambda|\lambda}(\theta)$
\begin{align}
 &= 
\int_{\Theta} \frac{G_{\lambda}(\alpha)}{D_{\lambda}(\alpha)} \,
S_{\mu}(\theta,\alpha) de_{\lambda}(\alpha) d \alpha \\
&+ \int_0^{\lambda} \int_{\Theta}
\frac{G_{\mu}(\alpha)}{D_{\mu}(\alpha)}
\left[\int_{\Theta} F_{\theta}[\Phi_{\lambda,\mu}(\theta,\beta)] 
S_{\mu}(\beta,\alpha)d \beta \right]
d e_{\mu}(\alpha) d \alpha \nonumber \\
&= \int_{\Theta} \frac{G_{\lambda}(\alpha)}{D_{\lambda}(\alpha)} \,
S_{\lambda}(\alpha,\theta) de_{\lambda}(\alpha) d \alpha
+ F_{\theta} [\widehat{x}_{\lambda | \lambda}(\theta)] d\lambda \\
&= F_{\theta} [\widehat{x}_{\lambda|\lambda}(\theta)]d\lambda + K_{\theta} [de_{\lambda}(\theta)] \,,
\end{align}
where $K_{\theta}$ is the integral transform defined as in (\ref{eq:k_int}).

\noindent
\textbf{Step 5. [Differential Equation for $S_{\lambda}(\alpha,\theta)$]}
The error covariance $S_{\lambda}(\alpha,\theta)$ can be written as
\begin{align}
S_{\lambda}(\alpha,\theta) &= 
E[\widetilde{x}_{\lambda | \lambda}(\alpha)\widetilde{x}_{\lambda | \lambda}(\theta)] \\
&= E[x_{\lambda}(\alpha)x_{\lambda}(\theta)] - 
E[\widehat{x}_{\lambda | \lambda}(\alpha)
\widehat{x}_{\lambda | \lambda}(\theta)] \label{eq:ss}\,.
\end{align}
Using expressions for $x_{\lambda}(\alpha)$ in (\ref{eq:x_l_l}), we can show that for $P_{\lambda}(\alpha,\beta) = E[x_{\lambda}(\alpha)x_{\lambda}(\theta)]$,
\begin{equation}
\frac{\partial P_{\lambda}(\alpha,\theta)}{d\lambda} = 
F_{\alpha}[P_{\lambda}(\alpha,\theta)] + F_{\theta} P_{\lambda}(\alpha,\theta) + C_{\lambda}(\alpha,\theta) \,.
\label{eq:p_l}
\end{equation}
Using the expression for $\widehat{x}_{\lambda|\lambda}(\alpha)$ in (\ref{eq:x_hat}), it can be shown that
\begin{align}
\frac{\partial E[\widehat{x}_{\lambda | \lambda}(\alpha)
\widehat{x}_{\lambda | \lambda}(\theta)]}{d \lambda} &= 
F_{\alpha}[E[\widehat{x}_{\lambda | \lambda}(\alpha)
\widehat{x}_{\lambda | \lambda}(\theta)]] \nonumber \\
&+ F_{\theta}[E[\widehat{x}_{\lambda | \lambda}(\alpha)
\widehat{x}_{\lambda | \lambda}(\theta)]] \label{eq:l_l}\\
&+ \int_{\Theta} \frac{G_{\lambda}^2(\beta)}{D_{\lambda}^2(\beta)} \, S_{\lambda}(\alpha,\beta) S_{\lambda}(\theta,\beta) d\beta \,.
\nonumber 
\end{align}
Differentiating (\ref{eq:ss}) and using (\ref{eq:p_l}) and (\ref{eq:l_l}), we get the desired equation:
\begin{align}
\frac{\partial}{\partial \lambda}S_{\lambda}(\alpha,\theta) &= 
F_{\alpha} [S_{\lambda}(\alpha,\theta)] + 
F_{\theta} [S_{\lambda}(\alpha,\theta)] + 
C_{\lambda}(\theta,\alpha) \nonumber \\
& -\int_{\Theta} \frac{G_{\lambda}^2(\beta)}{D_{\lambda}^2(\beta)} 
S_{\lambda}(\alpha,\beta) S_{\lambda}(\theta,\beta) d\beta
\,.
\end{align}

\section{Proof of Theorem~\ref{thm:recursive_smoother}: Recursive Smoother}
\label{app:proof_smoothing}
We now derive smoothing equations.  Using similar steps as in Lemma \ref{lemma:f_s}, we can show that
\begin{align}
\widehat{x}_{\lambda|T}(\theta) &= \int_{0}^{1}\int_{\Theta} g_{\lambda,\mu}(\theta,\alpha) d e_{\mu}(\alpha) d\alpha \,, \label{eq:sm_1}\\
g_{\lambda,\mu}(\theta,\alpha) &= \frac{\partial }{\partial \mu} E[x_{\lambda}(\theta) e_{\mu}(\alpha)] \,. \label{eq:sm_11}
\end{align}
Define the error covariance $S_{\lambda,\mu}(\theta,\alpha)$ as
\begin{equation}
S_{\lambda,\mu}(\theta,\alpha) = E[\widetilde{x}_{\lambda|\lambda}(\theta)
\widetilde{x}_{\mu|\mu}(\alpha)] \,.
\end{equation}
We have the following result for the smoother:
\begin{lemma}
The smoothed estimator $\widehat{x}_{\lambda|T}(\theta)$ is given by
\begin{equation}
\widehat{x}_{\lambda|T}(\theta) = \widehat{x}_{\lambda|\lambda}(\theta) + \int_{\lambda}^{1} \int_{\Theta} g_{\lambda,\mu}(\theta,\alpha) d e_{\mu}(\alpha) d\alpha \,, \label{eq:sm_2}
\end{equation}
where for $\mu \ge \lambda$,
\begin{equation}
g_{\lambda,\mu}(\theta,\alpha) = \frac{G_{\mu}(\alpha)}{D_{\mu}(\alpha)} S_{\lambda,\mu}(\theta,\alpha) \,. \label{eq:sm_3}
\end{equation}
\end{lemma}
\begin{IEEEproof}
Equation (\ref{eq:sm_2}) immediately follows from (\ref{eq:sm_1}) and Lemma \ref{lemma:f_s}. Equation (\ref{eq:sm_3}) follows by using (\ref{eq:innovation}) to compute $g_{\lambda,\mu}(\theta,\alpha)$ in (\ref{eq:sm_11}).
\end{IEEEproof}

We now want to characterize the error covariance $S_{\lambda,\mu}(\theta,\alpha)$.  Subtracting the telescoping representation in (\ref{eq:telescoping_arbit}) and the filtering equation in (\ref{eq:filter_field}), we get the following equation for the filtering error covariance:
\begin{equation}
d \widetilde{x}_{\mu|\mu}(\alpha) = \widetilde{F}_{\theta}[\widetilde{x}_{\mu|\mu}(\alpha)] d\mu
+ B_{\mu}(\alpha)dw_{\mu}(\alpha) - K_{\alpha}[dn_{\mu}(\alpha)] \,, \label{eq:error_covariance_eqn}
\end{equation}
where $\widetilde{F}_{\alpha}$ is the integral transform
\begin{equation}
\widetilde{F}_{\alpha}[\widetilde{x}_{\mu|\mu}(\alpha)] = F_{\alpha}[\widetilde{x}_{\mu|\mu}(\alpha)] - K_{\alpha}\left[\frac{G_{\mu}(\alpha)}{D_{\mu}(\alpha)}\widetilde{x}_{\mu|\mu}(\alpha) \right] \,.
\end{equation}
Just like we did in Lemma~\ref{lemma:solution}, we can write a solution to (\ref{eq:error_covariance_eqn}) as
\begin{align}
&\widetilde{x}_{\mu|\mu}(\alpha) = 
\int_{\Theta} \widetilde{\Phi}_{\mu,\lambda}(\alpha,\theta) 
\widetilde{x}_{\lambda|\lambda}(\theta) d\theta \nonumber \\
&+
\int_{\Theta} \int_{\lambda}^{\mu} \widetilde{\Phi}_{\mu,\gamma}(\alpha,\theta)[B_{\gamma}(\theta)dw_{\gamma}(\theta) - K_{\theta}[dn_{\gamma}(\theta)]] d\theta 
\label{eq:solution_x_tilde}
\end{align}
\begin{align}
\frac{\partial}{\partial \mu} \widetilde{\Phi}_{\mu,\lambda}(\alpha,\theta) = \widetilde{F}_{\alpha} \widetilde{\Phi}_{\mu,\lambda}(\alpha,\theta) 
\;{ \text{and} } \;
\widetilde{\Phi}_{\mu,\mu}(\alpha,\theta) = \delta(\alpha-\theta) \,.
\end{align}
Differentiating (\ref{eq:solution_x_tilde}) with respect to $\lambda$, we can show that
\begin{equation}
\int_{\Theta} \frac{\partial }{\partial \lambda} \widetilde{\Phi}_{\mu,\lambda}(\alpha,\beta) \widetilde{x}_{\lambda|\lambda}(\beta) d\beta \!=\!  -\!\int_{\Theta}
\widetilde{\Phi}_{\mu,\lambda}(\alpha,\beta) \widetilde{F}_{\beta} 
\widetilde{x}_{\lambda|\lambda}(\beta) d\beta \,. \label{eq:ddd}
\end{equation}
Substituting (\ref{eq:solution_x_tilde}) in (\ref{eq:error_covariance_eqn}), we have the following relationship:
\begin{equation}
S_{\lambda,\mu}(\theta,\alpha) = \int_{\Theta} \widetilde{\Phi}_{\mu,\lambda}(\alpha,\beta) S_{\lambda}(\beta,\theta) d\beta \,. \label{eq:ddd1}
\end{equation}

Substituting (\ref{eq:ddd1}) in (\ref{eq:sm_2}) and (\ref{eq:sm_3}), differentiating (\ref{eq:sm_2}) and using (\ref{eq:ddd}) and (\ref{eq:ricatti_s}), we get the following equation:
\begin{align}
&d \widehat{x}_{\lambda|T}(\theta) = F_{\theta} [\widehat{x}_{\lambda|T}(\theta)] d\lambda \nonumber \\
&+ \int_{\lambda}^{1}\int_{\Theta} \frac{G_u(\alpha)}{D_u(\alpha)}
\int_{\Theta} \widetilde{\Phi}_{\mu,\lambda}(\alpha,\beta) C_{\lambda}(\beta,\theta) d\beta d e_{\mu}(\alpha) d\alpha \,.
\label{eq:dc_1}
\end{align}
Assuming $S_{\lambda}(\theta,\theta) > 0$, we get smoother equations using the following calculations:
\begin{align}
&d \widehat{x}_{\lambda|T}(\beta)\delta(\theta-\beta) = F_{\beta} [\widehat{x}_{\lambda|T}(\beta)] \delta(\theta-\beta)d\lambda \label{eq:dc_2} \\
&\hspace{1cm} + \int_{\lambda}^{1}\int_{\Theta} \frac{G_u(\alpha)}{D_u(\alpha)}
\widetilde{\Phi}_{\mu,\lambda}(\alpha,\beta) C_{\lambda}(\beta,\theta) d e_{\mu}(\alpha) d\alpha \nonumber\\
&\frac{S_{\lambda}(\beta,\theta)}{C_{\lambda}(\beta,\theta)}
d \widehat{x}_{\lambda|T}(\beta)\delta(\theta-\beta) \nonumber \\
&= \frac{S_{\lambda}(\beta,\theta)}{C_{\lambda}(\beta,\theta)}
F_{\beta} [\widehat{x}_{\lambda|T}(\beta)] \delta(\theta-\beta)d\lambda \nonumber \\
&+ \int_{\lambda}^{1}\int_{\Theta} \frac{G_u(\alpha)}{D_u(\alpha)}
\widetilde{\Phi}_{\mu,\lambda}(\alpha,\beta) S_{\lambda}(\beta,\theta) d e_{\mu}(\alpha) d\alpha \label{eq:dc_3}\\
&\frac{S_{\lambda}(\theta,\theta)}{C_{\lambda}(\theta,\theta)}
d \widehat{x}_{\lambda|T}(\theta) = 
\frac{S_{\lambda}(\theta,\theta)}{C_{\lambda}(\theta,\theta)}
F_{\theta} [\widehat{x}_{\lambda|T}(\theta)] d\lambda \nonumber \\
&+ \quad \int_{\lambda}^{1}\int_{\Theta} \frac{G_u(\alpha)}{D_u(\alpha)}
\left(\int_{\Theta} \widetilde{\Phi}_{\mu,\lambda}(\alpha,\beta) S_{\lambda}(\beta,\theta) d\beta \right) d e_{\mu}(\alpha) d\alpha
\label{eq:dc_4} \\
&d \widehat{x}_{\lambda|T}(\theta) = 
F_{\theta} [\widehat{x}_{\lambda|T}(\theta)] d\lambda \nonumber \\
&\hspace{0.7cm}+\frac{C_{\lambda}(\theta,\theta)}{S_{\lambda}(\theta,\theta)}\int_{\lambda}^{1}\int_{\Theta} \frac{G_u(\alpha)}{D_u(\alpha)}
S_{\lambda,\mu}(\theta,\alpha) d e_{\mu}(\alpha) d\alpha \label{eq:dc_5}\\
&d \widehat{x}_{\lambda|T}(\theta) = 
F_{\theta} [\widehat{x}_{\lambda|T}(\theta)] d\lambda + 
\frac{C_{\lambda}(\theta,\theta)}{S_{\lambda}(\theta,\theta)}
[\widehat{x}_{\lambda|T}(\theta) -\widehat{x}_{\lambda|\lambda}(\theta)] \,. \label{eq:dc_6}
\end{align}
Equation (\ref{eq:dc_1}) is equivalent to (\ref{eq:dc_2}).  We multiply (\ref{eq:dc_2}) by $S_{\lambda}(\beta,\theta)/C_{\lambda}(\beta,\theta)$ to get (\ref{eq:dc_3}).  We integrate (\ref{eq:dc_3}) for all $\beta$ to get (\ref{eq:dc_4}).  To go from (\ref{eq:dc_4}) to (\ref{eq:dc_5}), we use (\ref{eq:ddd1}).  Equation (\ref{eq:dc_5}) follows from (\ref{eq:sm_2}).

To derive a differential equation for $S_{\lambda|T}(\alpha,\theta)$, we first note that
\begin{align}
S_{\lambda|T}(\alpha,\theta) &= 
E[\widetilde{x}_{\lambda|T}(\alpha)\widetilde{x}_{\lambda|T}(\theta)] \\
&= E[x_{\lambda}(\alpha) x_{\lambda}(\theta)] - 
E[\widehat{x}_{\lambda|T}(\alpha)\widehat{x}_{\lambda|T}(\theta)] \,.
\end{align}
Using (\ref{eq:sm_2}) to compute $E[\widehat{x}_{\lambda|T}(\alpha)\widehat{x}_{\lambda|T}(\theta)]$, we can find an expression for $S_{\lambda|T}(\alpha,\theta)$ as
\begin{align}
S_{\lambda|T}(\alpha,\theta) &= S_{\lambda}(\alpha,\theta) \nonumber \\
&\hspace{-1cm} -\int_{\lambda}^{1} \int_{\Theta} 
\frac{G_{\mu}^2(\alpha)}{D_{\mu}^2(\alpha)}
S_{\mu,\lambda}(\alpha_1,\alpha) S_{\mu,\lambda}(\alpha_1,\theta) d\mu d\alpha_1 \,. \label{eq:last}
\end{align}
Taking derivative of (\ref{eq:last}), we get (\ref{eq:smoother_cov}).

\section*{Acknowledgment}
The authors would like to thank the anonymous reviewers for their comments and suggestions, which greatly improved the quality and presentation of the paper.



\begin{IEEEbiographynophoto}{Divyanshu Vats}
(S'03) received the B.S. degree in electrical engineering and mathematics from The University of Texas at Austin in 2006.  He is currently working towards a Ph.D. in electrical and computer engineering at Carnegie Mellon University.

His research interests include detection-estimation Theory, probability and stochastic processes, information theory,
control theory, graphical models, and machine learning.
\end{IEEEbiographynophoto}

\begin{IEEEbiographynophoto}{Jos\'{e} M. F. Moura}
(S'71--M'75--SM'90--F'94) received degrees from Instituto Superior T\'ecnico (IST),
Lisbon, Portugal and from the Massachusetts Institute of Technology (MIT), Cambridge, MA. 
He  is University Professor at Carnegie Mellon University (CMU), having been on the faculty of IST and having held visiting faculty appointments at MIT. He manages a large education and research program between CMU and Portugal, www.icti.cmu.edu.
His research interests include statistical and algebraic signal and image processing, distributed inference, and network science. He published over  400  technical Journal and Conference papers, is the co-editor of two books, holds eight patents, and has given numerous invited seminars at international conferences, US and European Universities, and industrial and government Laboratories.

Dr.~Moura is \emph{Division Director Elect} (2011) of the IEEE, was the \emph{President} (2008-09) of the \emph{IEEE Signal Processing Society}(SPS),  \emph{Editor in Chief} for the {\em IEEE Transactions in Signal Processing}, interim \emph{Editor in Chief} for the \emph{IEEE Signal Processing Letters}, and was on the Editorial Board of several Journals, including the \emph{IEEE Proceedings}, the \emph{IEEE Signal Processing Magazine}, and the ACM \emph{Transactions on Sensor Networks}. He was on the steering and technical committees of several Conferences. Dr.~Moura is a \emph{Fellow} of the \emph{IEEE}, a \emph{Fellow} of the \emph{American Association for the Advancement of Science} (AAAS), and a corresponding member of the {\em Academy of Sciences of Portugal} (Section of Sciences). He was awarded the \emph{IEEE Signal Processing Society Meritorious Service Award}, the \emph{IEEE Millennium Medal}, an IBM Faculty Award, the CMU's College of Engineering \emph{Outstanding Research Award}, and the CMU Philip L.~Dowd Fellowship Award for Contributions to Engineering Education. In 2010, he was elected University Professor.
\end{IEEEbiographynophoto}

\end{document}